\documentclass[a4paper,reqno,11pt]{article}

\usepackage[hmargin=2.2cm,vmargin=2.2cm]{geometry}

\usepackage{amsmath,amssymb,amsthm,mathtools,mathrsfs,mathdots,color,graphicx,framed,eucal,authblk,enumerate}

\usepackage[colorlinks=true, pdfstartview=FitV, urlcolor=blue, citecolor=red, linkcolor=blue]{hyperref}

\usepackage[utf8]{inputenc}
\usepackage[T1]{fontenc}

\def\Im{\mathrm {Im}\,}
\def\N{\mathbb{N}}
\def\Z{\mathbb{Z}}

\def\Q{\mathbb{Q}}
\def\wt{\widetilde}
\def\wh{\widehat}

\def\1{\mathbf{1}}
\def\d{\mathrm d}
\def\e{\mathrm{e}}
\def\i{\mathrm{i}}
\def\pa{\partial}
\newcommand{\tr}[1][]{\mathrm{tr}_{#1}\,}

\def\B{\mathcal{B}}

\def\be{\begin{equation}}
\def\ee{\end{equation}}

\newtheorem{theorem}{Theorem}[section]

\newtheorem{definition}[theorem]{Definition}

\newtheorem{lemma}[theorem]{Lemma}
\newtheorem{proposition}[theorem]{Proposition} 
\newtheorem{corollary}[theorem]{Corollary}

\theoremstyle{remark}
\newtheorem{remark}[theorem]{Remark}

\def\p{\mathbf{p}}
\def\u{\mathbf{u}}
\def\DR{\mathrm{DR}}
\def\End{\mathrm{End}\,}
\def\SL{\mathrm{SL}}
\newcommand{\sltwo}{\mathfrak{sl}_2}
\def\Mgn{\overline{\mathcal M}_{g,n}}
\def\L{\Lambda}

\newcommand{\QtoP}{\Q^{\partitions}}
\newcommand{\partitions}{\mathscr{P}}
\newcommand{\pdv}[2]{\frac{\partial #1}{\partial #2}}
\newcommand{\blok}{\odot}
\newcommand{\bigblok}{\bigodot}

\def\fq#1{\left\langle{#1}\right\rangle_{\! q}}
\def\opq#1{\left\lbrace{#1}\right\rbrace_{\!q}}

\def\opg{\mathsf{Op}_c(g)}
\def\op#1{\mathsf{Op}_c\bigl({#1}\bigr)}
\def\opgzero{\mathsf{Op}_0(g)}
\def\opzero#1{\mathsf{Op}_0\bigl({#1}\bigr)}
\def\oplarge#1{\mathsf{Op}_c\biggl({#1}\biggr)}

\renewcommand{\=}{\: =\: }
\newcommand{\defis}{\: :=\: }
\newcommand{\+}{\,+\,}
\newcommand{\meno}{\,-\,}

\begin{document}

\numberwithin{equation}{section}

\title{Quantum KdV hierarchy and quasimodular forms}
\author[1,2]{Jan-Willem M. van Ittersum}
\author[3,4]{Giulio Ruzza}
\renewcommand\Affilfont{\footnotesize}
\affil[1]{\textit{Max-Planck-Institut f\"ur Mathematik, Vivatsgasse 7, 53111 Bonn, Germany;}}
\affil[2]{Current address: \emph{Department of Mathematics and Computer Science, University of Cologne,
Weyertal 86-90, 50931 Cologne, Germany;} {\texttt{j.w.ittersum@uni-koeln.de}}
}
\affil[3]{\textit{Grupo de F\'{i}sica Matem\'{a}tica, Departamento de Matemática, Instituto Superior T\'ecnico, Av. Rovisco Pais, 1049-001 Lisboa, Portugal;}}
\affil[4]{\textit{Departamento de Matem\'atica, Faculdade de Ci\^encias da Universidade de Lisboa, Campo Grande Edif\'{i}cio C6, 1749-016, Lisboa, Portugal;} \texttt{gruzza@fc.ul.pt}}

\date{}
\maketitle

\begin{abstract}
Dubrovin~\cite{Dubrovin} has shown that the spectrum of the quantization (with respect to the first Poisson structure) of the dispersionless Korteweg--de Vries (KdV) hierarchy is given by shifted symmetric functions; the latter are related by the Bloch--Okounkov Theorem~\cite{BlochOkounkov} to quasimodular forms on the full modular group.
We extend the relation to quasimodular forms to the full quantum KdV hierarchy (and to the more general quantum Intermediate Long Wave hierarchy).
These quantum integrable hierarchies have been defined by Buryak and Rossi~\cite{BuryakRossi2} in terms of the Double Ramification cycle in the moduli space of curves.
The main tool and conceptual contribution of the paper is a general effective criterion for quasimodularity.
\end{abstract}

\medskip
\medskip



\medskip

\section{Introduction}\label{sec:intro}

\subsection{Differential polynomials and \texorpdfstring{$q$}{q}-series.}\label{sec:1.1}

Let $\L:=\Q[\p]$ be the ring of polynomials with rational coefficients in the variables~$p_j\,,$ for $j\geq 1$, collectively denoted by $\p:=(p_1,p_2,p_3,\dots)$.
Assigning weight~$k$ to $p_k$ we have the grading $\L = \bigoplus_{n\geq 0}\L_n\hspace{1pt},$ where $\L_n$ consists of polynomials of homogeneous weight~$n$.
For any linear operator $G\in\End(\L)$ such that\footnote{Throughout this paper we denote by $[A,B] := AB-BA$ the commutator of~$A$ and~$B$.}
\be
\label{eq:condition}
\biggl[G\, , \,\sum_{k\geq 1}k\, p_k\,\pdv{}{p_k}\biggr]\=0,
\ee
i.e., such that $G$ restricts to a linear operator on~$\L_n$ for any $n\geq 0$, we introduce the~$q$-series
\be
\label{eq:qseries}
\opq G \defis \frac{\sum_{n\geq 0}q^n\,\tr[\L_n]G}{\sum_{n\geq 0}q^n\,\dim\L_n}.
\ee
Equivalently, $\opq G = q^{-1/24}\eta(q)\sum_{n\geq 0}q^n\,\tr[\L_n]G$, where $\eta(q) = q^{1/24}\prod_{k\geq 1}(1-q^k)$ is the Dedekind eta function.

In certain cases, the~$q$-series~\eqref{eq:qseries} specializes to the well-studied~{$q$-bracket}, introduced in \cite{BlochOkounkov}.
The latter is attached to a function $f:\partitions\to\Q$ from the set~$\partitions$ of partitions to the rationals and defined by
\be
\label{eq:BOqseries}
\fq{f}\defis\frac{\sum_{\lambda\in\partitions}q^{|\lambda|}f(\lambda)}{\sum_{\lambda\in\partitions}q^{|\lambda|}}\=q^{-1/24}\eta(q)\sum_{\lambda\in\partitions}q^{|\lambda|}f(\lambda),
\ee
where $|\lambda|:=\lambda_1+\dots+\lambda_\ell$ denotes the integer $\lambda=(\lambda_1,\ldots,\lambda_\ell)$ is a partition of. (Namely, if $G$ is the diagonal operator acting as multiplication by $f(\lambda)$ on $p_\lambda=p_{\lambda_1}\cdots p_{\lambda_\ell}\hspace{1pt},$ then $\opq G=\fq f\hspace{1pt}.$)

In particular, Bloch and Okounkov \cite{BlochOkounkov} proved that for the class of \emph{shifted symmetric} functions of partitions, the~{$q$-bracket}~\eqref{eq:BOqseries} is quasimodular of homogeneous weight.
This class consists of homogeneous polynomials of certain basic functions $Q_k:\partitions\to\Q$, for $k\geq 0$, where the weight of~$Q_k$ is defined to be~$k$.
The basic functions are defined by $Q_0(\lambda)=1$ and 
\be
\label{eq:Qk}
Q_k(\lambda) \defis
\frac 1{(k-1)!}\sum_{i\geq 1}\left[\left(\lambda_i-i+\tfrac{1}{2}\right)^{k-1}-\left(-i+\tfrac{1}{2}\right)^{k-1}\right]\+\beta_k
\ee
for $k\geq 1$, where $\beta_k = \left(\frac 1{2^{k-1}}-1\right)\frac{B_k}{k!}$ and $B_k$ denotes the $k$th Bernoulli number.
The central characters of the symmetric group are shifted symmetric functions \cite{KerovOlshanski} and hence these functions appear in the study of asymptotic properties of partitions \cite{IvanovOlshanski,Rios}, as well as in many works in enumerative geometry, e.g., in the Hurwitz/Gromov--Witten theory of an elliptic curve \cite{Dijkgraaf, Engel, EskinOkounkov, HahnIttersumLeid, Ochiai, OkounkovPandharipande}, or in the determination of the Siegel--Veech constants of the moduli space of flat surfaces \cite{EskinOkounkov,ChenMollerZagier,ChenMollerSauvagetZagier}.

The appearance of shifted symmetric functions in the study of integrable hierarchies is part of an interesting story.
As explained in \cite{Dubrovin}, Eliashberg \cite{Eliashberg} solved the quantization problem for the classical Hopf hierarchy by using ideas coming from \emph{Symplectic Field Theory}.
Concretely, he constructed a commuting family of quantum Hamiltonians~$G_k^{\sf Hopf}$ ($k\geq -2$), which are operators on~$\Lambda$ depending on a constant\footnote{This constant is denoted $u_0$ in \cite{Dubrovin}.}~$c$, 
obtained from differential polynomials by the procedure explained in the next paragraph.
In particular, after inserting 
a \emph{quantization} parameter $\hbar$ in these differential polynomials (see Remark~\ref{remark:normalization}), in the limit $\hbar \to 0$ they reduce to the Hamiltonian densities of the classical Hopf hierarchy.
Rossi \cite{Rossi} showed that the operators~$G_k^{\sf Hopf}$, under the \emph{boson-fermion correspondence} (see, e.g., \cite{MiwaJimboDate}), are quadratic in fermions---a fact that was exploited by Dubrovin to diagonalize these operators (to then provide applications to the symplectic field theory of the disk). Namely, \cite[Theorem~1.4]{Dubrovin}
\be
\label{eq:Dubrovin}
G_k^{\sf Hopf}\, s_\lambda(\p)=E_{k}^{[0]}(\lambda)\,s_\lambda(\p),\qquad \lambda\in\partitions,\quad k\geq -2,
\ee
where the eigenvalues~$E_k^{[0]}:\partitions\to\Q$ are shifted symmetric functions, given explicitly by
\be
\label{eq:disperionlesseigenvalue}
E_k^{[0]}\=\sum_{j=0}^{k+2}\frac{c^{k+2-j}}{(k+2-j)!}\,Q_j\hspace{1pt},
\ee
and the eigenvectors~$s_\lambda(\p)$ are the Schur functions\footnote{Expressed in terms of the power sum polynomials.} \cite{Macdonald}, defined by
\be \label{eq:Schur}
s_\lambda(\p) \defis \det\left[h_{\lambda_i-i+j}(\p)\right]_{i,j=1}^{\ell(\lambda)},\qquad \sum_{k\in\Z} y^k h_k(\p) \= \exp\Bigl(\sum\nolimits_{k\geq 1}\frac{p_k}ky^k\Bigr).
\ee
It follows immediately from~\eqref{eq:Dubrovin} and the Bloch--Okounkov Theorem that
\be \label{eq:opqEliashberg}
\opq {G_k^{\sf Hopf}} \= \sum_{j=0}^{k+2}\frac{c^{k+2-j}}{(k+2-j)!}\,\fq {Q_j}
\ee
is a polynomial in $c$ with quasimodular coefficients.
Note that this expression is of homogeneous weight~$k+2$, provided we assign weight~$+1$ to~$c$, and consider the quasimodular coefficients with the corresponding quasimodular weight.

In the present work we study quasimodularity properties of a deformation  $G_k^{\sf KdV}(\epsilon)$ (constructed by Buryak and Rossi~\cite{BuryakRossi2}) of the operators $G_k^{\sf Hopf}$, depending (polynomially) on an additional parameter~$\epsilon$ and satisfying $G_k^{\sf KdV}(0)=G_k^{\sf Hopf}$.
These are the Hamiltonian operators of the \emph{quantum Korteweg--de Vries} (KdV) hierarchy\footnote{The quantization is with respect to the first Poisson structure, cf.~\cite{BuryakRossi2,Dubrovin}.}, as the corresponding densities (again, after properly introducing the parameter $\hbar$, see Remark~\ref{remark:normalization}) reduce in the limit $\hbar\to 0$ to the Hamiltonian densities of the classical KdV hierarchy.
(Incidentally, the construction in op.\ cit., which we briefly review in Section~\ref{sec:BR}, is much more general, and produces a quantum integrable hierarchy attached to any Cohomological Field Theory. The KdV case corresponds to the trivial Cohomological Field Theory.)

Our first goal is to study whether the quasimodularity of $q$-series associated with the Hopf hierarchy, expressed by~\eqref{eq:opqEliashberg}, survives under the deformation in~$\epsilon$ (as anticipated in~\cite{RuzzaYang}).
We answer this in the affirmative (Theorem~\ref{thm:KdV}), by providing a general criterion for quasimodularity (Theorem~\ref{thm:main}).
Moreover, as a byproduct of the general criterion, we obtain a simplification of the quantum Hamiltonian operators which we expect to be useful in the study and classification of quantum integrable hierarchies of rank~$1$ \cite{BuryakDubrovinGuereRossi}.

We now move on to a more detailed explanation of our findings.

\paragraph{From differential polynomials to operators.} 
In this paper, we shall be concerned with quasimodular properties of $\opq G$ for operators~$G$ on~$\L$ which are obtained out of a polynomial $g\in\Q[\u]$ by the following \emph{quantization} procedure.
Here $\Q[\u]$ is the ring of polynomials with rational coefficients in the variables~$u_j\hspace{1pt},$ for $j\geq 0$, collectively denoted by $\u:=(u_0,u_1,u_2,\dots)$.

First, for $j\in\Z_{\geq 0}\hspace{1pt},$ we introduce the Fourier series\footnote{With the convention $0^0 = 1$.} $v_j(x):=\sum_{\ell\in\Z}(\i \ell)^j\,\omega_\ell\,\e^{\i \ell x}$, where we assign weight $\ell$ to $\omega_\ell\hspace{1pt}.$
Note that $v_j(x)=\pa_x^j v_0(x)$ for $j\geq 1$.
Next, given $g\in\Q[\u]$, define a formal power series (of homogeneous weight $0$) in the variables $\omega_\ell\hspace{1pt},$ $\ell\in\Z$, by
\be
\int_0^{2\pi}g(\mathbf v(x))\,\frac{\d x}{2\pi},\qquad \mathbf v(x) \defis (v_0(x),v_1(x),\dots).
\ee
Write all monomials in this series as a product of $\omega_\ell$'s where all the $\omega_\ell$'s with $\ell\geq 0$ appear to the left of all $\omega_\ell$'s with $\ell<0$ (\emph{normal ordering}).
In this expression, we finally replace $\omega_\ell$ with the operator $\widehat{\omega}_\ell\in\End(\Lambda)$, defined by
\be
\label{eq:representation}
\widehat{\omega}_\ell \,f \defis 
\begin{cases}
p_\ell \, f&\ell\geq 1,
\\[2pt]
c \, f&\ell=0,
\\[2pt]
-\ell\,\frac{\partial f}{\partial p_{-\ell}},&\ell\leq-1,
\end{cases}\qquad\qquad (\ell\in \Z,f\in\L)
\ee
with $c\in\Q$ an arbitrary constant.
We denote by $\opg$ the operator on~$\L$ obtained from~$g\in\Q[\u]$ in this way. 
In other words, we give the following definition.

\begin{definition}
For $g\in \Q[\u]$, denote by $\opg:\L\to \L\oplus \L\i$ the operator, depending on~$c\in \Q$, given by
\[ \opg \defis \int_0^{2\pi}:\mspace{1mu}g\bigl(\widehat{\mathbf{v}}(x)\bigr)\mspace{1mu}:\frac{\d x}{2\pi}\,,
\]
where $:\mspace{1mu}\cdots\mspace{1mu}:$ denotes the normal ordering, $\widehat{\mathbf{v}}(x) = \bigl(\widehat{v}_0(x),\widehat{v}_1(x),\ldots)$ with $\widehat{v}_j(x):=\sum_{\ell\in \Z} (\mathrm{i}\ell)^j \, \widehat{\omega}_\ell \,e^{\mathrm{i}\ell x}$ and $\widehat{\omega}_\ell$ as defined above. 
\end{definition}

\label{pageevenodd}It is worth noting that $\opg$ might be complex-valued.
More precisely, let $\Q[\u] = \Q[\u]^{\sf even}\oplus\Q[\u]^{\sf odd},$ where $\Q[\u]^{\sf even}$ (respectively, $\Q[\u]^{\sf odd}$) is the span of monomials which are even (respectively, odd) with respect to the weight operator~$\sum_{j\geq 1}j\mspace{1mu}u_j\mspace{1mu}\pdv{}{u_j}$.
Then, $\opg$ is purely real for $g\in\Q[\u]^{\sf even}$, and purely imaginary for $g\in\Q[\u]^{\sf odd}$.\vspace{2pt}

Let us give a few examples:\vspace{-7pt}
\begin{itemize}\itemsep2pt
\item $\opg =0$ when~$g$ is in the image of the operator $\displaystyle \pa_x\defis \sum\nolimits_{j\geq 0}u_{j+1}\pdv{}{u_j},$ \vspace{-5pt}
\item $\op {u_0} \= c$,
\item $\displaystyle \op {u_0^2} \= c^2\+2\sum\nolimits_{j\geq 1}j \,p_j\,\frac{\partial}{\partial p_j}$,
\item $\displaystyle \op{u_0^3} \= c^3\+6c\sum\nolimits_{j\geq 1}j \, p_j\, \frac{\partial}{\partial p_j}\+6\Delta,$ where
\be
\Delta \= \frac 12\sum_{j,k\geq 1}\Bigl((j+k)\,p_j\,p_k\,\pdv{}{p_{j+k}}+j\,k\,p_{j+k}\,\pdv{^2}{p_j\pa p_k}\Bigr)
\ee
is the \emph{cut-and-join} operator \cite{Goulden}.
\end{itemize}

If $g\in\Q[\u]^{\sf even}$, the operator~$\opg$ is symmetric (see \cite[Lemma~2.3]{RuzzaYang}), i.e., $(v,\opg w)=(\opg v,w)$ for all $v,w\in\L$ with respect to the standard scalar product $(\ ,\,)$ on~$\L$ (see \cite{Macdonald}).
The latter can be defined by
\be \label{eq:innerproduct}
(p_{\lambda},p_\mu) \= z_\lambda\delta_{\lambda,\mu}\hspace{1pt},
\ee
where
\be 
\label{eq:monomial}
p_\lambda \defis p_{\lambda_1}\cdots p_{\lambda_\ell}
\ee
is the monomial basis of~$\Lambda$ indexed by partitions $\lambda = (\lambda_1,\dots,\lambda_\ell)$, and $z_\lambda := \prod_{m\geq 1} r_m(\lambda)!\,m^{r_m(\lambda)} $, where $r_m(\lambda):=\#\{i \mid \lambda_i=m\}$.\vspace{2pt}

\subsection{Quantum Korteweg--de Vries hierarchy.}
A construction by Buryak and Rossi \cite{Buryak,BuryakRossi1,BuryakRossi2}, also inspired by previous works in Symplectic Field Theory \cite{Eliashberg,EliashbergGiventalHofer}, provides an effective construction of quantum integrable hierarchies associated with an arbitrary Cohomological Field Theory (CohFT).
Even though the construction is completely general, in this work we restrict to the case of rank~$1$ CohFTs only.

The output of this construction, which is briefly reviewed in Section~\ref{sec:BR}, is a family of \emph{Hamiltonian densities} $g_k(\mathbf u;\epsilon)\in\Q[\u]^{\sf even}\otimes\Q[\![\epsilon]\!]$, for $k\geq -2$, possibly depending on the parameters of the CohFT (more details below).
They are determined by either an explicit formula in terms of the Double Ramification cycles in the moduli space of curves, see~\eqref{eq:defg}, or (more effectively) by a recurrence relation of order one, see~\eqref{eq:recursionBR1} and~\eqref{eq:recursionBR2}.

One of the main properties of the \emph{Hamiltonian operators} $G_k(\epsilon):=\op{g_k(\u;\epsilon)}\in\End(\L)\otimes\Q[\![\epsilon]\!]$ is that they enjoy the commutativity\footnote{It is appropriate to remark that the setting of \cite{BuryakRossi2} is more general, and the only requirement is the commutation relation $[\wh\omega_a,\wh\omega_b] = -a\delta_{a+b,0}$. For our purposes it is convenient to fix the representation~\eqref{eq:representation} (see also~\cite{Dubrovin}).}
\be
[G_j(\epsilon),G_k(\epsilon)] = 0,\qquad j,k\geq -2.
\ee

The relevance of this construction to the theory of integrable systems stems from the fact that, after suitably introducing a quantization parameter~$\hbar$ (see Remark~\ref{remark:normalization}), the Hamiltonian densities reduce, in the limit $\hbar\to 0$, to those of \emph{classical} integrable hierarchies, well-studied in the literature.
We refer to the aforementioned literature, in particular to the Introduction of \cite{Dubrovin}, for more details on this point.
Since in this work we are mainly interested in the quantum Hamiltonian densities, we opted to simplify the exposition by dropping the parameter~$\hbar$ (which can be reinstated at any time by the transformations described in Remark~\ref{remark:normalization}).

The simplest instance of this construction is provided by the quantum \emph{Korteweg--de Vries} (KdV) hierarchy (associated with the trivial CohFT), a prototypical example of an integrable system.
In this case, the first few densities read
\begin{align}
\nonumber
g_{-2}^{\sf KdV}(\u;\epsilon) \= &1,\qquad
g_{-1}^{\sf KdV}(\u;\epsilon)\=u_0\hspace{1pt},\qquad
g_{0}^{\sf KdV}(\u;\epsilon)\=\frac{u_0^2}2-\frac 1{24}\+\frac\epsilon{24} u_2\hspace{1pt},
\\ \nonumber
g_{1}^{\sf KdV}(\u;\epsilon) \= &\frac{u_0^3}6-\frac {u_0}{24}-\frac {u_2}{24}\+\frac\epsilon{24}\Bigr(u_0u_2-\frac 1{120}\Bigr)\+\left(\frac\epsilon{24}\right)^{\! 2}\frac{u_4}2,
\\ \nonumber
g_{2}^{\sf KdV}(\u;\epsilon) \= &\frac{u_0^4}{24}-\frac{u_0^2}{48}-\frac{u_0u_2}{24}+\frac 7{5760}\+\frac\epsilon{24}\left(
\frac {u_0^2u_2}2 -\frac{u_4}{30}-\frac {u_2}{24}-\frac {u_0}{120}\right) \+
\\
&\+\left(\frac\epsilon{24}\right)^{\! 2}\left(\frac{7 u_2^2}{10}+\frac{u_0u_4}2-\frac1{210}\right)\+\left(\frac\epsilon{24}\right)^{\! 3}\frac {u_6}6.
\end{align}
(For the recursion determining them see~\eqref{eq:recursionBR1} and~\eqref{eq:recursionBR2} below, and for further properties see Appendix~\ref{app}.)

The Hamiltonian operators $G_k^{\sf KdV}(\epsilon):=\op{g_k^{\sf KdV}(\u;\epsilon)}$ are polynomials of degree~$k$ in~$\epsilon$.
As already mentioned in Section~\ref{sec:1.1}, their constant term $G_k^{[0]}:= [\epsilon^0]G_k^{\sf KdV}(\epsilon)$ coincides with the Hamiltonian operators of the Hopf hierarchy, i.e., $G_k^{[0]}=G_k^{\sf Hopf}$, and the quasimodularity of $\bigl\{G_k^{[0]}\bigr\}_{\!q}$ follows by Dubrovin's result~\eqref{eq:Dubrovin} along with the Bloch--Okounkov theorem. On the other hand, their leading coefficient
$G_k^{[\infty]}:=[\epsilon^k]G_k^{\sf KdV}(\epsilon)$ is given by (see Corollary~\ref{corollary:infinitedispersion})
\be
\label{eq:Gkinf}
G_k^{[\infty]} \= \frac{c^2}2\delta_{k,0}\+
\frac {L_{2k+2}}{(-4)^k(2k+1)!!},\qquad k\geq 0,
\ee
where the operator $L_k$ is given by
\be
\label{eq:operatorLk}
L_k \defis -\frac{c^2}2\delta_{k,2}\meno\frac{B_{k}}{2k}-\frac{\mathrm{i}^{k}}2\op{u_0u_{k-2}}\=
-\frac{B_{k}}{2k}\+\sum_{j\geq 1}j^{k-1}\,p_j\pdv{}{p_j}.
\ee
Therefore, $G_k^{[\infty]}$ are diagonal on the monomial basis~\eqref{eq:monomial} of~$\Lambda$,
\be
G_k^{[\infty]}\,p_\lambda \= E^{[\infty]}_k(\lambda)\,p_\lambda,\qquad \lambda\in\partitions,\quad k\geq -2,
\ee
where $E^{[\infty]}_k:\partitions\to\Q$ are given in terms of the \emph{moment functions}~\cite{Zagier}
\be\label{eq:momentfunctions}
S_k(\lambda)\defis -\frac{B_k}{2k} \+ \sum_{i=1}^{\ell(\lambda)}\lambda_i^{k-1},\qquad k\geq 1,
\ee
by 
\be
\label{eq:infiniteepseigenvalue}
E^{[\infty]}_k \= \frac{c^2}{2}\delta_{k,0} \+ \frac{S_{2k+2}}{(-4)^k(2k+1)!!},\qquad k\geq 0.
\ee
It was observed in~\cite{Zagier} that the~{$q$-bracket} of~$S_{2k+2}$ is an Eisenstein series which is quasimodular of weight~${2k+2}$.

Let us denote by $\wt M$ the ring of quasimodular forms (with rational coefficients) on the full modular group~$\SL_2(\Z)$ (see e.g., \cite[Section~5.3]{Zagier123}), containing, for each even $k\geq 2$ the Eisenstein series
\be
\label{eq:Eisenstein}
\mathbb{G}_k\=-\frac{B_k}{2k} \+ \sum_{m,r\geq 1} m^{k-1}q^{mr}.
\ee
Then $\wt M = \bigoplus \wt M_k$ is a graded ring freely generated over the rationals by the \emph{Eisenstein series}~$\mathbb G_2\hspace{1pt},$ $\mathbb G_4$ and~$\mathbb G_6\hspace{1pt},$ where the weight of $\mathbb G_k$ is defined to be $k$.
Moreover, let $\wt M[c,\epsilon] := \wt M \otimes \Q[c,\epsilon] =: \bigoplus _k \wt M[c,\epsilon]_k\hspace{1pt},$ where we assign weight~$+1$ to~$c$ and~$-1$ to~$\epsilon$.

\begin{framed}
\begin{theorem}
\label{thm:KdV}
For the quantum KdV Hamiltonian operators 
\be
G_k^{\sf KdV}(\epsilon) \defis \op{g_k^{\sf KdV}(\u;\epsilon)}\in\End(\L)\otimes\Q[c,\epsilon], \qquad (k\geq -2)
\ee
we have
\be
\bigl\lbrace G_k^{\sf KdV}(\epsilon)\bigr\rbrace_{\! q} \,\in\,\wt M[c,\epsilon]_{k+2}\,.
\ee
\end{theorem}
\end{framed}

The proof is given in Section~\ref{sec:proof2}, for the more general case of the quantum \emph{Intermediate Long Wave} hierarchy, which is a generalization of the KdV hierarchy \cite{BuryakRossi2} (see Theorem~\ref{thm:ILW}).
The key step in the proof is a general criterion (Theorem~\ref{thm:main} below) for quasimodularity of homogeneous weight which applies to operators of the form~$\opg$ for $g\in\Q[\u]$.
Explicit examples for $k\leq 3$ in terms of Eisenstein series are included in Appendix~\ref{app2}.

We expect that Theorem~\ref{thm:KdV} holds true for all rank~$1$ quantum Double Ramification integrable hierarchies.
Namely, the quantum KdV Hamiltonian densities are the special case $s_i=0$ of a more general hierarchy of Hamiltonian densities $g_k(\u;\epsilon,\mathbf s)$, depending on parameters $\mathbf s = (s_1,s_3,s_5,\ldots)$ (see, Section~\ref{sec:BR}).
There exist different possible normalizations for these Hamiltonian densities (see, \cite[eq.~5.3]{BuryakDubrovinGuereRossi}), and by Theorem~\ref{thm:main} below in all normalizations $\opq{\op{g_k(\epsilon,\mathbf s)}}$ belongs to $\wt M[c][\![\mathbf s,\epsilon]\!]$.
We expect that there exists a convenient normalization such that $\opq{\op{g_k(\epsilon,\mathbf s)}}$ is quasimodular \emph{of homogeneous weight} (where $s_k\,,$ for $k$ odd, is assigned weight $k$).

Moreover, as suggested to us by Don Zagier, we expect that the $q$-series~\eqref{eq:qseries} associated to arbitrary compositions of the quantum KdV operators give rise to quasimodular forms as well, or, even stronger, that the eigenvalues of the quantum KdV operators are shifted symmetric functions of homogeneous weight.
Namely, by~\cite{RuzzaYang} there exists a simultaneous basis~$r_\lambda(\p;\epsilon)\in\L_{|\lambda|}[\![\epsilon]\!]$ of eigenfunctions~$E_k(\lambda;\epsilon)$ for~$G_k^{\sf KdV}(\epsilon)$ for all $\lambda\in\partitions$. 
 Then, we expect that
\begin{align}\label{conj}
E_k(\lambda;\epsilon) \,\in\, \Q[c,Q_2,Q_3,\ldots][\![\epsilon]\!]_{k+2}\,,
\end{align}
where $\epsilon$ is assigned weight $-1$, $c$ weight $+1$, and $Q_k$ weight $k$.
Note $E_k(\lambda;0) = E_k^{[0]}(\lambda)$ is the shifted symmetric function in~\eqref{eq:disperionlesseigenvalue} and, as a consequence of Theorem~\ref{thm:KdV}, we have
$\langle E_k(\lambda;\epsilon) \rangle_q = \bigl\lbrace G_k^{\sf KdV}(\epsilon)\bigr\rbrace_{\! q} \in \wt M[c,\epsilon]_{k+2}\,.$ We have numerical evidence for~\eqref{conj} in a few instances, and we hope to return to it in a later publication.

\subsection{A criterion for homogeneous quasimodularity}

We will show that the~$q$-series~$\opq{\opg}$ (for $g\in \Q[\u]$) is always a quasimodular form \emph{of mixed weight}.
Moreover, we provide a criterion for the quasimodularity of \emph{homogeneous weight} for the~$q$-series~$\opq{\opg}\hspace{1pt}.$
To state it, we assign weight~${k+1}$ to $u_k\hspace{1pt},$ so that $\Q[\u]$ (as well as its subspaces~$\Q[\u]^{\sf even}$ and~$\Q[\u]^{\sf odd}$ defined at page \pageref{pageevenodd}) become graded algebras\footnote{It might seem more natural to write $u_{k+1}$ for what is called $u_k$ here, but in order to adhere to standard notation for the Hamiltonian densities introduced in the next section, we do not do so. In fact, recall that in the background there is a second weight operator, which gives rise to the even and odd part of $\Q[\u]$, and assigns weight $k$ to $u_k\,$. 
}. Moreover, let $\wt M[c] := \wt M \otimes \Q[c] =:\bigoplus_{k}\wt M[c]_k\hspace{1pt},$ be the polynomial ring in~$c$ and~$\epsilon$ over the graded ring of quasimodular forms, graded by the quasimodular weight and by assigning weight~$+1$ to~$c$.

Moreover, let~$\B$ be the linear operator on~$\Q[\u]$ defined by
\be
\label{eq:B}
\B \defis \exp\Bigl(-\frac 12\sum_{{i,j\geq 0}} (-1)^{\frac{i-j}{2}}\,\frac{B_{i+j+2}}{i+j+2}\,\pdv{^2}{u_i\partial u_j}\Bigr),\qquad B_k = k\text{th Bernoulli number.}
\ee
\begin{framed}
\begin{theorem}\label{thm:main}\mbox{}\\[-20pt]
\begin{enumerate}[{\upshape (i)}]\itemsep0pt
\item  For any $g\in\Q[\u]$ we have $\opq{\opg} \in \wt M[c]$, 
\item For any $g\in\Q[\u]^{\sf odd}$, we have $\opq{\opg}=0$,
\item The mapping $\Q[\u] \to \widetilde{M}[c]$
\begin{align}\label{eq:morphism}
g &\mapsto \left\lbrace\op{\B g}\right\rbrace_{\! q}
\end{align}
is a morphism of graded vector spaces.
\end{enumerate}
\end{theorem}
\end{framed}

The proof is given in Section~\ref{sec:proof1} and builds on the previous work \cite{vanIttersumSymmetric} of the first author. In the special case~${c=0}$, the theorem states the following.
\begin{corollary} Let $g\in \Q[\u]$ be such that $\B^{-1} g$ is of homogeneous weight. Then, $\opq{\opgzero}$ is quasimodular of homogeneous weight. 
\end{corollary}

\paragraph{Holomorphic anomaly equation.}
Just as in \cite{vanIttersum1} we answer the question when $\opq{\opgzero}$ is actually modular (rather than quasimodular).
Namely, if we, instead, consider~$c$ to be a formal variable, the \emph{holomorphic anomaly equation} of~$\opq{\opg}$ (determining the failure of modularity) can be expressed as
\be
\label{eq:hae}
-2\,\mathfrak{d} \opq{\opg} \= \pdv{^2}{c^2} \opq{\opg}\hspace{1pt},
\ee
(cf.~Proposition~\ref{prop:HAE}) where~$\mathfrak{d}$ is the unique derivation on quasimodular forms which vanishes on modular forms and for which $\mathfrak{d}(\mathbb{G}_2)=-\frac{1}{2}$, where $\mathbb{G}_2=-\frac{1}{24}+\sum_{m,r\geq 1} m\,q^{mr}$ is the Eisenstein series of weight~$2$.
Together with the differential operator~$q\pdv{}{q}$ and the weight operator, this derivation~$\mathfrak{d}$ gives an action of~$\sltwo$ on quasimodular forms.
Note that~\eqref{eq:hae} can equivalently be written as $-2\mathfrak{d} \opq{\opg} =  \lbrace\op{\partial^2 g/\partial u_0^2}\rbrace_{\!q}\hspace{1pt}$ (cf.~Lemma~\ref{lemma:simple}). 
Hence, $\opq{\opgzero}$ is modular precisely if $\lbrace\opzero{\partial^2 g/\partial u_0^2}\rbrace_{\!q}=0$.

\begin{remark}
\label{remark:hae}
By~\eqref{eq:disperionlesseigenvalue} the holomorphic anomaly equation~\eqref{eq:hae} can be explicitly checked in the limit $\epsilon\to 0$, because we know by \cite[Theorem~3]{Zagier} that $\mathfrak{d}\fq{Q_j}=-\frac{1}{2}\fq{Q_{j-2}}$ for all $j\geq 2$.
In the limit $\epsilon\to\infty$, the holomorphic anomaly equation~\eqref{eq:hae} is consistent with~\eqref{eq:infiniteepseigenvalue}, because, as mentioned earlier, in~\cite{Zagier} it is shown that the~{$q$-bracket} of~$S_{2k+2}$ is an Eisenstein series of weight~${2k+2}$, and every Eisenstein series of weight $2k+2$ with $k>0$ is actually modular. 
\end{remark}

\paragraph{Functions on partitions having the same $q$-bracket.}
The~{$q$-bracket}~\eqref{eq:BOqseries} is not only a specialization of the~$q$-series~\eqref{eq:qseries}; it is also related to this $q$-series by the following construction.
For any basis~$b_\lambda(\p)$ of~$\L$, indexed by partitions $\lambda\in\partitions$ and satisfying $b_\lambda\in\Lambda_{|\lambda|}\hspace{1pt},$ and for any $G\in\End(\L)$ satisfying~\eqref{eq:condition} we have the equality $\opq G = \fq f$ where~$f(\lambda)$ is defined to be the~$\lambda$th diagonal entry of the matrix representation of~$G$, i.e., $f(\lambda) := (b_\lambda,Gb_\lambda)/(b_\lambda,b_\lambda)$, where the inner product is defined by~\eqref{eq:innerproduct}. Hence, by studying $\opq{G}$ for $G\in \End(\L)$, we study the $q$-brackets~$\fq{f}$ for many functions~$f$ at the same time. 

Incidentally, this reflects the fact that different functions~$f:\partitions\to\Q$ can have the same~{$q$-bracket}.
For example, as observed in \cite[Section~13]{ChenMollerZagier} the moment functions~\eqref{eq:momentfunctions}
and the shifted symmetric functions\footnote{These shifted symmetric functions have a nice interpretation as the moments of the hook-lengths of partitions, see \cite[Section~13]{ChenMollerZagier}.} (in terms of the generators $Q_k\hspace{1pt},$ defined by~\eqref{eq:Qk})
\be
T_k(\lambda) \defis \frac{(k-2)!}{2} \sum_{i=0}^{k} (-1)^i \, Q_i(\lambda) \, Q_{k-i}(\lambda),\qquad k\geq 2,
\ee
are two instances of functions on partitions for which, in case $k$ is even, the $q$-bracket equals the Eisenstein series~\eqref{eq:Eisenstein}. 
This particular example can be explained as $T_k$ is the so-called \emph{M\"oller transform} of \cite{Zagier} of $S_k\hspace{1pt}.$ Now, the M\"oller transform corresponds to the change of coordinates between the monomial basis~$p_\lambda$ and the Schur basis~$s_\lambda$ of~$\Lambda$.
Indeed, we have\footnote{The second equality follows from \cite[Section~13]{ChenMollerZagier} after expanding $s_\lambda$ in the monomial basis of~$\Lambda$.}
\be
S_{k}(\lambda) = \frac{(p_\lambda,L_{k}p_\lambda)}{(p_\lambda,p_\lambda)}, \qquad T_{k}(\lambda) = \frac{(s_\lambda,L_ks_\lambda)}{(s_\lambda,s_\lambda)}, \qquad k\geq 2 \text{ even},
\ee
where $L_k$ is the operator defined in~\eqref{eq:operatorLk}.

\subsection*{Outline of the rest of the paper.}
In Section~\ref{sec:proof1} we prove Theorem~\ref{thm:main} and the holomorphic anomaly equation~\eqref{eq:hae}.
In Section~\ref{sec:BR} we review the construction of Double Ramification quantum integrable hierarchies.
In Section~\ref{sec:proof2} we prove Theorem~\ref{thm:ILW} which is a generalization of Theorem~\ref{thm:KdV}.
In Appendix~\ref{app} we give explicit formulas for the limits $\epsilon\to 0,\infty$ of the quantum KdV Hamiltonian densities.
In Appendix~\ref{app2} we illustrate our two main theorems by giving tables of quasimodular forms obtained as $\lbrace G_{k}^{\sf KdV}(\epsilon)\rbrace_{\! q}$ and as $\opq{\opg}$ for some $g\in \Q[\u]$. 

\section{Partitions and quasimodular forms}\label{sec:proof1}
\subsection{Proof of Theorem~\ref{thm:main}}
First of all, we refine the $q$-bracket~\eqref{eq:BOqseries} as follows. Define the~\emph{$\underline{x}$-bracket}\footnote{In \cite{vanIttersumSymmetric} this was called the $\underline{u}$-bracket.} of a function $f:\partitions\to \Q$ by 	
\be
\langle f\rangle_{\underline{x}} \= \frac{\sum_{\lambda\in \partitions} f(\lambda)\, x_\lambda }{\sum_{\lambda \in \partitions} x_\lambda}\,\in\, \Q[\![x_1,x_2,x_3,\ldots]\!]
\quad\quad\quad (x_\lambda=x_{\lambda_1} x_{\lambda_2}\cdots) .
\ee
Then, one has~$\fq{f}=\langle f \rangle_{(q, q^2, q^3,\ldots)}$.
Observe that the~$\underline{x}$-bracket defines an isomorphism of vector spaces  (but not of algebras!)
\be\QtoP \xrightarrow{\sim} \Q[\![x_1,x_2,x_3,\ldots]\!], \quad \quad f \mapsto \langle f \rangle_{\underline{x}}\, .\ee
Given $f,g\in \QtoP$ we define their \emph{induced product}~${f\blok g}$ by
\be \label{eq:inducedprod}\langle f\blok g \rangle_{\underline{x}}\= \langle f\rangle_{\underline{x}}\,\langle g\rangle_{\underline{x}}\, ,\ee
where the product of~$\langle f\rangle_{\underline{x}}$ and~$\langle g\rangle_{\underline{x}}$ is the usual product of power series. In particular, $\fq{f\blok g}=\fq{f}\fq{g}\hspace{1pt}.$ 

Now, given a partition $\lambda$, write $r_m(\lambda)=\#\{i \mid \lambda_i=m\}$, and similarly for $\mathbf{k}\in (\Z_{>0})^n$ we write $r_m(\mathbf{k}) = \#\{i \mid k_i=m\}$. Note (for example, using \cite[Proposition 3.1.4]{vanIttersumSymmetric}) that
\begin{align}\label{eq:rmxbrac}
\langle r_m \rangle_{\underline{x}} = \sum_{r=1}^\infty x_m^r = \frac{x_m}{1-x_m}.
\end{align}
In the proof of Theorem~\ref{thm:main}, we will need the following result.
\begin{lemma}\label{lem:1}
Given $n\geq 0$ and $\mathbf{k}\in (\Z_{>0})^n$, we have
\be \Bigl\langle \prod_{m\geq 1} \binom{r_m(\lambda)}{r_m(\mathbf{k})}\Bigr\rangle_{\!q} \= \prod_{k\in \mathbf{k}}\frac{q^k}{1-q^k}.
\ee
\end{lemma}
\begin{proof}
By \cite[Proposition 7.2.3(ii)]{vanIttersumSymmetric} we have
\be \prod_{m\geq 1} \binom{r_m(\lambda)}{r_m(\mathbf{k})} \= \bigblok_{k\in \mathbf{k}}  r_k(\lambda). \ee
The result follows from \eqref{eq:inducedprod} and \eqref{eq:rmxbrac}.
\end{proof}

\begin{proof}[Proof of Theorem~\ref{thm:main}]
Recall that the $p_\lambda=p_{\lambda_1}\cdots p_{\lambda_r}$ for $\lambda \in \partitions$ form a basis for~$\Lambda$. 
The main observation is that with respect to this basis, for all $\mathbf{k}\in \N^r, \mathbf{l}\in \N^s$ and $\lambda \in \partitions$ we have
\begin{align} [p_\lambda]\,p_{k_1}\cdots p_{k_r}&\, \pdv{^n}{p_{l_1}\cdots \partial p_{l_s}} p_\lambda \=  
 \begin{cases} \prod_m\binom{r_m(\lambda)}{r_m(\mathbf{l})}\, r_m(\mathbf{l})!& \mathbf{k} \text{ is a permutation of } \mathbf{l}, \\
0 & \text{else,} \end{cases}
 \end{align}
where~$[p_\lambda]$ indicates we extract the coefficient of~$p_\lambda\hspace{1pt}.$ Here, $r_m(\lambda)$ and $r_m(\mathbf{l})$ are as defined above. 

Given a monomial $u_{\mathbf{a}} = u_{a_1}\cdots u_{a_n}\in \Q[\u]$, consider
\be
\op{u_{\mathbf{a}}} \=  \sum_{\substack{\mathbf{k}\in \Z^n\\ |\mathbf{k}|=0}} (\mathrm{i}\mathbf{k})^{\mathbf{a}}\,c^{\#\{j \mid k_j=0\}}\prod_{j\mid k_j>0}p_{k_j}\prod_{j \mid k_j<0}\biggl(-k_j\pdv{}{p_{-k_j}}\biggr),
\ee
where $(\mathrm{i}\mathbf{k})^{\mathbf{a}} = \prod_{j=1}^n (\mathrm{i}k_j)^{a_j}$. As the diagonal contribution of this operator is zero unless we have the equality of multisets $\{k_j \mid k_j>0\}=\{-k_j \mid k_j<0\}$, we can partition the set of indices in pairs. That is, write~$\Pi_2(n,m)$ for the partitions of all $B\subset\{1,2,\ldots,n\}$ in $m$~sets of two elements, i.e., $\pi\in \Pi_2(n,m)$ can be written as $\pi=\{A_1,\ldots,A_{m}\}$ with $|A_i|=2$ for all~$i$ and $\bigcup A_i=B$ for some $B\subset\{1,\ldots,n\}$ with $|B|=2m$. Often, we write $B(\pi)$ for $B$ to stress the dependence of $B$ on $\pi$. 
 Observe that $|\Pi_2(n,m)|=\binom{n}{2m}\cdot(2m-1)!!$. For a set~$S$, write $a_S=\sum_{i\in S} a_i$ and $s(\mathbf{a},S)=\i^{a_S}\sum_{i\in S} (-1)^{a_i}$. We write $|\mathbf{a}|$ for $a_{\{1,\ldots,n\}} = \sum_{i=1}^n a_i\hspace{1pt}.$
Also, for $\mathbf{k}\in \Z^m$ write $z_{\mathbf{k}} = \prod_{m} r_m(\mathbf{k})!\, m^{r_m(\mathbf{k})}$.
Then, 
\begin{align}
\nonumber
[p_\lambda]\,\op{u_{\mathbf{a}}}\, p_\lambda &\= 
[p_\lambda]\sum_{m=0}^{\lfloor n/2 \rfloor}c^{n-2m}\!\sum_{\substack{\pi\in \Pi_2(n,m)\\ a_{B(\pi)} = |\mathbf{a}|}} \sum_{k_{A_1},\ldots,k_{A_m}\geq 1}\! \frac{1}{z_{\mathbf{k}}}\prod_{A\in \pi} \left(s(\mathbf{a},A)\,k_A^{a_{A}+2}\,p_{k_A}\right)\prod_{A\in \pi}\pdv{}{p_{k_A}} p_{\lambda}\\
&\=\sum_{m=0}^{\lfloor n/2 \rfloor}c^{n-2m}\sum_{\substack{\pi\in \Pi_2(n,m)\\ a_{B(\pi)} = |\mathbf{a}|}}\sum_{k_{A_1},\ldots,k_{A_m}\geq 1} \prod_{A\in \pi} s(\mathbf{a},A)\,k_A^{a_{A}+1}\,\prod_{m} \binom{r_m(\lambda)}{r_m(\mathbf{k})}.
\end{align}
From Lemma~\ref{lem:1}, it follows that
\begin{align} 
\nonumber
\opq{\op{u_{\mathbf{a}}}} 
&\=
\sum_{m=0}^{\lfloor n/2 \rfloor}c^{n-2m}\sum_{\substack{\pi\in \Pi_2(n,m)\\ a_{B(\pi)} = |\mathbf{a}|}}\sum_{k_{A_1},\ldots,k_{A_m}\geq 1} \prod_{A\in \pi} s(\mathbf{a},A)\,k_A^{a_{A}+1}\,  \frac{q^{k_A}}{1-q^{k_A}}\\
&\= \label{eq:trace2}
\sum_{m=0}^{\lfloor n/2 \rfloor}c^{n-2m}\sum_{\substack{\pi\in \Pi_2(n,m)\\ a_{B(\pi)} = |\mathbf{a}|}} \prod_{A\in \pi}\,s(\mathbf{a},A) \Bigl(\frac{B_k}{2k}+\mathbb{G}_{a_A+2}\Bigr),
\end{align}
where $\mathbb{G}_k$ for $k\geq 2$ is the holomorphic Eisenstein series~\eqref{eq:Eisenstein} of weight~$k$. 

Observe that if $\mathbf{a}=(i,j)$ for some $i,j\in\Z$, then 
\be
s(\mathbf{a},\{1,2\})
\=(-1)^{\frac{i+j}{2}}\bigl((-1)^{i}+(-1)^j\bigr)
\=
\begin{cases} 2(-1)^{\frac{i-j}{2}} & i\equiv j \mod 2\\
0 & \text{else.}
\end{cases}
\ee
Hence, $s(\mathbf{a},A)$ vanishes if ${a_A+2}$ is odd. As the Eisenstein series~$\mathbb{G}_k$ are quasimodular for even~$k$, we have shown that $\opq{\opg}\in \wt M[c]$ for all $g\in \Q[\u]$. 

Next, assume $u_{\mathbf{a}}\in \Q[\u]^{\sf odd}$, i.e., $|\mathbf{a}|$ is odd. Then, in~\eqref{eq:trace2} we see that every summand of $\opq{\op{u_{\mathbf{a}}}}$ contains a factor $s(\mathbf{a},A)$ for which $a_A$ is odd. Hence, such a factor $s(\mathbf{a},A)$ vanishes.
That is, $\opq{\op{u_{\mathbf{a}}}}=0$. Therefore, by linearity, $\opq{\opg}=0$ for all $g\in \Q[\u]^{\sf odd}$. 

By definition of~$\B$ we deduce from~\eqref{eq:trace2} that
\begin{align}\label{eq:opasQMF}
\left\lbrace\op{\B u_{\mathbf{a}}} \right\rbrace_{\! q} \=\sum_{m=0}^{\lfloor n/2 \rfloor}c^{n-2m}\sum_{\substack{\pi\in \Pi_2(n,m)\\ a_{B(\pi)} = |\mathbf{a}|}}\prod_{A\in \pi} s(\mathbf{a},A)\, \mathbb{G}_{a_A+2}\hspace{1pt},
\end{align}
which is of homogeneous weight $|\mathbf{a}|+n$ (where $n$ is the length of the vector~$\mathbf{a}$). 
Therefore, it follows that $g \mapsto \lbrace\op{\B g}\rbrace_{\! q}$ is a morphism of graded vector spaces. 
\end{proof}
Observe that as every quasimodular form has a (highly non-unique) representation as a polynomial in Eisenstein series, it follows that the mapping $g \mapsto \lbrace\opzero{\B g}\rbrace_{\! q}$ from $\Q[\u]$ to $\widetilde{M}$ is surjective.

\subsection{Holomorphic anomaly equation}
A quasimodular form is a holomorphic function $f=f(\tau)$ of $\tau$ in the complex upper half plane, admitting a Fourier series ($q$-series, $q=\e^{2\pi\i\tau}$) at infinity and such that for all $\left(\begin{smallmatrix} a & b \\ c & d\end{smallmatrix}\right)\in\SL_2(\Z)$
\begin{align}\label{eq:transfo} (c\tau+d)^{-k} f\Bigl(\frac{a\tau+b}{c\tau+d}\Bigr) \= \sum_{j=0}^p \frac{(\mathfrak{d}^jf)(\tau)}{j!} \Bigl(\frac{1}{2\pi i}\frac{c}{c\tau+d}\Bigr)^j,\qquad \Im \tau>0,
\end{align}
where $\mathfrak{d}$ is the derivation defined in the introduction.
The transformation~\eqref{eq:transfo} of the quasimodular forms in Theorem~\ref{thm:main} is determined by the holomorphic anomaly equation expressed in Proposition~\ref{prop:HAE} below.

We first have a simple lemma.

\begin{lemma}
\label{lemma:simple}
For any $g\in\Q[\u]$ we have $\pdv{}{c}\opg\=\op{\pdv{g}{u_0}}$.
\end{lemma}
\begin{proof}
For any $g(\u)$ the operator $\op{g(u_0-c,u_1,\dots)}$ is independent of~$c$ by construction.
Hence
\be 
\opg \= \sum_{s\geq 0}\frac{c^s}{s!}\,\op{(\pa_{u_0})^sg(u_0-c,u_1,\dots)}
\ee
and so
\begin{align}
\nonumber
\pdv{}{c}\,\opg \={}&\sum_{s\geq 1}\frac{c^{s-1}}{(s-1)!}\op{(\pa_{u_0})^sg(u_0-c,u_1,\dots)} 
\\
\={}& \sum_{s\geq 0}\frac{c^s}{s!}\op{(\pa_{u_0})^s(\pa_{u_0}g)(u_0-c,u_1,\dots)} = \oplarge{\pdv{g}{u_0}g},
\end{align}
as claimed.
\end{proof}

\begin{proposition}
\label{prop:HAE}
For all $g\in \Q[\u]$, we have
\be
 -2\,\mathfrak{d} \opq{\opg} \= \pdv{^2}{c^2} \opq{\opg} \= \opq{\oplarge{\pdv{^2g}{u_0^2}}}.
\ee
\end{proposition}
\begin{proof}
Recall that by~\eqref{eq:opasQMF} we have
\be 
\lbrace\op{\B u_{\mathbf{a}}} \rbrace_{\! q} \=\sum_{m=0}^{n/2}c^{n-2m}\sum_{\substack{\pi\in \Pi_2(n,m)\\ a_{B(\pi)} = |\mathbf{a}|}}\prod_{A\in \pi} s(\mathbf{a},A)\, \mathbb{G}_{a_A+2}\hspace{1pt}.
\ee
For the first equation, it suffices to show that $-2\,\mathfrak{d}$ and $\pdv{^2}{c^2}$ agree on $\lbrace\op{\B u_{\mathbf{a}}} \rbrace_{\! q}\,$. We compute
\be
-2\,\mathfrak{d} \lbrace\op{\B u_{\mathbf{a}}} \rbrace_{\! q}  \=2\sum_{m=0}^{\lfloor n/2\rfloor} c^{n-2m}\sum_{\substack{\pi\in \Pi_2(n,m)\\ a_{B(\pi)} = |\mathbf{a}|}}\sum_{\substack{A'\in \pi \\ a_{A'}=0}}\prod_{\substack{A\in \pi \\ A\neq A'}} s(\mathbf{a},A)\, \mathbb{G}_{a_A+2}\hspace{1pt},
\ee
where we used that $\mathfrak{d}$ is a derivation which annihilates $\mathbb{G}_{a_A+2}$ except if $a_A=0$. Now, let $\pi'\in \Pi_2(n,m-1)$ be the partition after removing $A'$ from $\pi$. Note that there are $\binom{n-2m+2}{2}$ partitions $\pi$ yielding $\pi'$. Hence, the first equality follows from the computation
\begin{align}
\nonumber
-2\,\mathfrak{d} \lbrace\op{\B u_{\mathbf{a}}} \rbrace_{\! q}  &\=2\sum_{m=0}^{\lfloor n/2\rfloor}c^{n-2m}\binom{n-2m+2}{2}\sum_{\substack{\pi'\in \Pi_2(n,m-1)\\ a_{B(\pi)} = |\mathbf{a}|}}\prod_{A\in \pi'} s(\mathbf{a},A)\, \mathbb{G}_{a_A+2} \\
&\=\pdv{^2}{c^2}\sum_{m=1}^{\lfloor n/2\rfloor}c^{n-2m+2}\sum_{\substack{\pi'\in \Pi_2(n,m-1)\\ a_{B(\pi)} = |\mathbf{a}|}}\prod_{A\in \pi'} s(\mathbf{a},A)\, \mathbb{G}_{a_A+2} \=\pdv{^2}{c^2}\lbrace\op{\B u_{\mathbf{a}}} \rbrace_{\! q} \hspace{1pt}.
\end{align}
The second equality follows from Lemma~\ref{lemma:simple}.
\end{proof}

\section{Applications to quantum integrable hierarchies}

\subsection{Double Ramification quantum integrable hierarchies}\label{sec:BR}

A Cohomological Field Theory (CohFT) \cite{KontsevichManin} consists of \vspace{-7pt}
\begin{itemize}\itemsep0pt
\item a finite-dimensional $\Q$-vector space $V$, equipped with a non-degenerate symmetric two-form $\eta\in{\rm Sym}^2(V^*)$ and with a distinguished element $\1\in V$, and 
\item linear maps $c_{g,n}:V^{\otimes n}\to H^{\sf even}(\Mgn,\Q)$, indexed by $g,n\geq 0$ such that $2g-2+n\geq 0$.
\end{itemize}
Here, $\Mgn$ is the Deligne--Mumford moduli space of stable curves of genus~$g$ with~$n$ marked points and $H^{\sf even}(\Mgn,\Q)$ is the even part of its rational cohomology ring.
The maps $c_{g,n}$ have to satisfy a number of axioms prescribing their behaviour under natural maps between the moduli spaces, i.e., under permutation of marked points, forgetting of marked points, and glueing of curves.
For more details see \cite[Section~3]{Buryak} and references therein.

In \cite{BuryakRossi2} a family of Hamiltonian densities $g_k(\u;\epsilon)$ is defined starting from an arbitrary CohFT.
We shall consider here only the case of one-dimensional CohFTs, namely $V = \Q\1$, under the additional assumption $\eta(\1\otimes \1)=1$ and, by a slight abuse of notation, we shall denote $c_{g,n}$ the value at $\1^{\otimes n}$.
A result of Teleman \cite{Teleman} implies that all such CohFTs are given by
\be
\label{eq:Teleman}
c_{g,n} \= \exp\Bigl(\sum\nolimits_{j\geq 1}s_{2j-1}\,{\rm ch}_{2j-1}(\mathbb E)\Bigr),
\ee
in terms of parameters $s_k\hspace{1pt},$ for $k\geq 1$ and odd, where ${\rm ch}_k(\mathbb E)\in H^k(\Mgn,\Q)$ are the Chern characters of the Hodge bundle $\mathbb E$ over $\Mgn$.
The densities are defined by\footnote{The density $g_{-2}=1$ is not introduced in \cite{BuryakRossi2}.}
\be
g_{-2}(\u;\epsilon) \defis 1,\qquad g_{-1}(\u;\epsilon) \defis u_0 \hspace{1pt},
\ee
and $g_k(\u;\epsilon)$ for $k\geq 0$ is defined in terms of the Fourier series (already used in the introduction) $v_j(x)= \sum_{\ell\in\Z}(\i\ell)^j\,\omega_\ell\,\e^{\i\ell x}$, for $j\geq 0$, by requiring that
\be
\label{eq:defg}
g_k(\mathbf v(x);\epsilon) \defis\!	 \sum_{\begin{smallmatrix}g,n\geq 0 \\ 2g-2+n\geq 0\end{smallmatrix}}\frac 1{n!}\sum_{(a_1,\dots,a_n)\in\Z^n}\int_{\DR_g(-|a|,a_1,\dots,a_n)}\psi_1^{k}\,\Lambda_g^\vee(\epsilon)\,c_{g,n+1}\,\omega_{a_1}\dots\omega_{a_n}\,\e^{\i|a| x},
\ee
where $|a|:=\sum_{i=1}^na_i$ and $\mathbf v(x) \defis (v_0(x),v_1(x),\dots)$, see \cite{BuryakRossi1,BuryakRossi2} for more details.
Here, we use the following standard notations (see, e.g., \cite{BuryakRossi2} and references therein for more details):\vspace{-7pt}
\begin{itemize}\itemsep0pt
\item
$\DR_g(a_0,a_1,\dots,a_n)\in H_{2g-2+n}(\overline{\mathcal M}_{g,n+1},\Q)$ is the \emph{Double Ramification cycle} (roughly speaking, defined as the locus in $\overline{\mathcal M}_{g,n+1}$ of stable curves with marked points $(C;p_0,\dots,p_n)$ such that $\mathcal O_C\simeq\mathcal O_C(\sum_{i=0}^na_ip_i)$),
\item
$\psi_1\in H^2(\overline{\mathcal M}_{g,n+1},\Q)$ is the Chern class of the cotangent line bundle at the first marked point, and
\item
$\Lambda_g^\vee(\epsilon) = 1-\epsilon\lambda_1+\dots+(-\epsilon)^g\lambda_g\hspace{1pt},$ where $\lambda_i\in H^{2i}(\overline{\mathcal M}_{g,n+1},\Q)$ are the Chern classes of the Hodge bundle.
\end{itemize}
It is worth pointing out that all densities~$g_k(\u;\epsilon)$ are sums of even monomials with respect to the weight operator $\sum_{j\geq 0}u_j\pdv{}{u_j}$, i.e., $g_k(\u;\epsilon)\in \Q[\u]^{\sf even}[\epsilon]$ (see \cite[Appendix~B]{BuryakRossi2}).

\begin{remark}
\label{remark:normalization}
To simplify the exposition we have omitted the quantization parameter $\hbar$ of \cite{BuryakRossi2}, as the normalization of op.\ cit.\ can be recovered from~\eqref{eq:defg} by the transformations
\be
\epsilon\mapsto\epsilon(\i\hbar)^{-\frac 12},\qquad s_k\mapsto s_k(\i\hbar)^{\frac k2},\qquad u_k\mapsto u_k(\i\hbar)^{-\frac 12},\qquad g_k\mapsto (\i\hbar)^{\frac{k+2}2}g_k\hspace{1pt}.
\ee
This follows directly from the dimensional constraints of the integrals over the Double Ramification cycle in~\eqref{eq:defg}.
Moreover, the parameter $\epsilon$ in this paper corresponds to $\varepsilon^2$ in \cite{BuryakRossi2}.
To compare with the normalization of \cite{Dubrovin}, where $\epsilon=0$, we need to replace $\i\hbar$ with $\hbar$.
\end{remark}

It is also follows from \cite[Theorem~3.5 and Lemma~3.7]{BuryakRossi2}, combined with Remark~\ref{remark:normalization} (see also \cite[Section~4]{RuzzaYang}), that the densities $g_k(\u;\epsilon)$ in~\eqref{eq:defg} can be determined from $g_{-1}(\u;\epsilon)=u_0$ by the recursion
\begin{align}
\label{eq:recursionBR1}
\pdv{g_{k+1}(\u;\epsilon)}{u_0}\= &g_k(\u;\epsilon),\\
\label{eq:recursionBR2}(k+2+\mathcal D)\pa_x g_{k+1}(\u;\epsilon)\= &\left[g_k(\u;\epsilon),\overline{g_1(\u;\epsilon)}\right],\qquad k\geq -1,
\end{align}
provided one has computed $g_1(\u;\epsilon)$ in advance, at least up to constant terms in $\u$, cf.~\eqref{eq:commutator}.
Here
\be
\label{eq:D}
\partial_x \defis \sum_{i\geq 0}u_{i+1}\pdv{}{u_i},\qquad\mathcal D \defis \epsilon\pdv{}{\epsilon}\+\sum_{i\geq 1}(2i-1)s_{2i-1}\pdv{}{s_{2i-1}},
\ee
and the expression $[f,\overline g]$ is defined for $f,g\in\Q[\u][\![\epsilon]\!]$ by (cf.~\cite[Equation~2.2]{BuryakRossi2})
\be
\label{eq:commutator}
[f,\overline g] \= \sum_{n\geq 1}\frac{(-1)^{n-1}}{n!}\!\sum_{\substack{r_1,\dots,r_n\geq 0 \\ s_1,\dots,s_n\geq 0}}\frac{\partial^nf}{\partial u_{s_1}\cdots\pa u_{s_n}}(-1)^{|\mathbf{r}|}P_{r_1+s_1+1,\dots,r_n+s_n+1}(\pa_x)\,\frac{\pa^ng}{\pa u_{r_1}\cdots\pa u_{r_n}},
\ee
and does not depend on constant terms of $f$ and $g$ in $\u$.
In~\eqref{eq:commutator}, $|\mathbf{r}|=\sum_{i=1}^n r_i$ and $P_{\ell_1,\dots,\ell_n}(\xi)$ are polynomials in $\xi$ defined by the sequence of their coefficients, namely,
\be
\label{eq:PwtP}
[\xi^j]\, P_{\ell_1,\dots,\ell_n}(\xi) \= \begin{cases}(-1)^{\frac{n-1-j+|\mathbf{\ell}|}2}\,[\xi^j]\wt P_{\ell_1,\dots,\ell_n}(\xi) & n-1-j+|\mathbf{\ell}|\mbox{ is even},\\
0 & \mbox{otherwise},\end{cases}
\ee
and $\wt P_{\ell_1,\dots,\ell_n}(\xi)$ are polynomials in $\xi$ determined by their values at positive integers $\xi$,
\be
\wt P_{\ell_1,\dots,\ell_n}(\xi) \= \sum_{\substack{a_1,\dots,a_n\geq 0 \\ a_1+\dots+a_n = \xi}}a_1^{\ell_1}\cdots a_n^{\ell_n}.
\ee

Later we shall need the following particular cases.

\begin{lemma}
\label{lemma:commutator}
We have $P_\ell(\xi) \= \xi^\ell$ and
\begin{align}
\nonumber
P_{\ell,m}(\xi) &\= \frac{\ell!\,m!}{(\ell+m+1)!}\xi^{\ell+m+1} \+
\\
\label{eq:Plm}
&\qquad\+\sum_{i\geq 0}\frac{B_{2i+2}}{2i+2}\,\biggl(\!(-1)^{\ell+i}\,{\binom{\ell}{2i+1-m}}\+(-1)^{m+i}\,\binom{m}{2i+1-\ell}\!\biggr)\,\xi^{\ell+m-2i-1}.
\end{align}
\end{lemma}
\begin{proof}(See also \cite[Lemma 6.1.2]{vanIttersumSymmetric}.)
Only~\eqref{eq:Plm} needs a proof.
We compute the generating series (for $\xi$ integer)
\be
\mathbf P(u,v;\xi)\defis\sum_{\ell,m\geq 0}\wt P_{\ell,m}(\xi)\frac{u^\ell}{\ell!}\frac{v^m}{m!}\=\sum_{a=0}^\xi\e^{au}\e^{(\xi-a)v} \= 
\e^{\xi v}\sum_{a=0}^\xi\e^{a(u-v)} \= \frac{\e^{u(\xi+1)}-\e^{v(\xi+1)}}{\e^u-\e^v}.
\ee
The proof is complete by Taylor expanding $\mathbf P(u,v;\xi) \= \frac{\e^{u\xi}-1}{\e^{u-v}-1}\+\frac{\e^{v\xi}-1}{\e^{v-u}-1}$ using $\frac z{\e^z-1} \= \sum_{j\geq 0} B_j\frac{z^{j}}{j!}$ and by~\eqref{eq:PwtP}.
\end{proof}

\begin{remark}
\label{remark:recursionwellposed}
It is explained in \cite[Section~3.5]{BuryakRossi2} that the recursion equations~\eqref{eq:recursionBR1} and ~\eqref{eq:recursionBR2} uniquely determine the $g_k$'s, for all $k\geq -1$, using as initial data $g_{-1}=u_0$ and $g_1$ (the latter up to additive constants which are fixed by the recursion).
Indeed, one first uses~\eqref{eq:recursionBR2} to determine the~$g_k$'s for $k\geq -1$ up to a constant depending on $k,\epsilon,\mathbf s$ only.
Note that this yet undetermined constant does not affect the right-hand side of~\eqref{eq:recursionBR2} so that the recursion works.
These constants are finally determined by~\eqref{eq:recursionBR1}.
\end{remark}

\subsection{Quantum Intermediate Long Wave hierarchy}\label{sec:proof2}

The quantum Intermediate Long Wave hierarchy corresponds to the construction of Buryak and Rossi for the \emph{Hodge} CohFT (see also \cite{BSTV2014})
\be
c_{g,n} \= 1+\mu\lambda_1+\dots+\mu^g\lambda_g\hspace{1pt},
\ee
where $\mu$ is a parameter and $\lambda_k\in H^{2k}(\Mgn,\Q)$ are the Chern classes of the Hodge bundle.
In terms of the parameters $\mathbf s$ in~\eqref{eq:Teleman} we have $s_{2i-1}=(2i-2)!\,\mu^{2i-1}$, whence~\eqref{eq:D} reduces to
\be
\label{eq:DILW}
\mathcal D \= \epsilon\pdv{}{\epsilon}\+\mu\pdv{}{\mu}.
\ee
Let us denote~$g_k^{\sf ILW}(\u;\epsilon,\mu)$ and~$G_k^{\sf ILW}(\epsilon,\mu):=\op{g_k^{\sf ILW}(\u;\epsilon,\mu)}$ the densities and operators for this hierarchy.
We know from \cite[Lemma~4.2]{BuryakRossi2} that\footnote{In loc.\ cit.\ the formula is given up to constant terms.}
\be
G_1^{\sf ILW}(\epsilon,\mu) \= \oplarge{\frac{u_0^3}{6}-\frac{u_0}{24}+(\epsilon-\mu)\sum_{g\geq 1}(\epsilon\mu)^{g-1}\frac{|B_{2g}|}{2(2g)!}\left(u_0u_{2g}-\frac{|B_{2g+2}|}{2g+2}\right)}\hspace{1pt}.
\ee

We start by making the recursion~\eqref{eq:recursionBR2} more explicit in this case.

\begin{lemma}
The operator
\be
\mathcal R^{\sf ILW}: g(\u;\epsilon,\mu)\mapsto\left[g(\u;\epsilon,\mu),\overline{g_1^{\sf ILW}(\u;\epsilon,\mu)}\right]
\ee
can be spelled out as
\be
\mathcal R^{\sf ILW} \= \mathcal R_1^{\sf ILW}+\mathcal R_2^{\sf ILW},
\ee
where
\begin{align}
\label{eq:RILW1}
\mathcal R_1^{\sf ILW} \= &\sum_{i\geq 0}\partial_x^{i+1}\Big(\frac{u_0^2}2+(\epsilon-\mu)\sum_{g\geq 1}(\epsilon\mu)^{g-1}\frac{|B_{2g}|}{(2g)!}u_{2g}\Big)\frac{\pa}{\pa u_i}-\frac 12\sum_{i,j\geq 0}\frac{(i+1)!(j+1)!}{(i+j+3)!}u_{i+j+3}\frac{\pa^2}{\pa u_i\pa u_j},
\\
\label{eq:RILW2}
\mathcal R_2^{\sf ILW} \= &\sum_{i,j,l\geq 0}\frac{B_{2l+2}}{2l+2}\,\biggl(\!(-1)^{i+l}\,\binom{i+1}{2l-j}+(-1)^{j+l}\binom{j+1}{2l-i}\!\biggr)\,u_{i+j+1-2l}\,\frac{\pa^2}{\pa u_i\pa u_j}.
\end{align}
\end{lemma}
\begin{proof}
It follows directly from~\eqref{eq:commutator} and Lemma~\ref{lemma:commutator}.
\end{proof}

Introduce the \emph{reduced} ILW densities
\be
\label{eq:reducedILW}
\wt g_k^{\sf ILW} \defis \B^{-1} g_k^{\sf ILW}\hspace{1pt},
\ee
where $\B$ is the operator given in~\eqref{eq:B}.
Remarkably, they satisfy a similar but slightly simpler recursion.

\begin{lemma}\label{lemma:simplifiedrecursion}
The reduced ILW densities $\wt g_{k}^{\sf ILW}(\u;\epsilon,\mu)$ are uniquely determined from $\wt g_{-2}^{\sf ILW}(\u;\epsilon,\mu)=1$ and $\wt g_{-1}^{\sf ILW}(\u;\epsilon,\mu)=u_0$ by the recursion
\begin{align}
\label{eq:recursionreduced1}
\pdv{\wt g_{k+1}^{\sf ILW}(\u;\epsilon,\mu)}{u_0}&\= \wt g_k^{\sf ILW}(\u;\epsilon,\mu),\\
\label{eq:recursionreduced2}
(k+2+\mathcal D)\pa_x\wt g_{k+1}^{\sf ILW}(\u;\epsilon,\mu)&\=\mathcal R_1^{\sf ILW}\,\wt g_k^{\sf ILW}(\u;\epsilon,\mu),\qquad k\geq -1,
\end{align}
where $\mathcal R_1^{\sf ILW}$ is given in~\eqref{eq:RILW1}.
\end{lemma}

\begin{proof}
For the sake of clarity, let us drop the superscript ${\sf ILW}$ and the dependence on $\u,\epsilon,\mu$ in this proof.
The equation~\eqref{eq:recursionreduced1} follows from the chain of equalities
\be
\frac{\pa}{\pa u_0}\wt g_{k+1}
=\B^{-1}\frac{\pa}{\pa u_0}g_{k+1}
=\B^{-1} g_k
=\wt g_k,
\ee
where we use $[\B,\partial_{u_0}]=0$.
Next, for any power series $\Phi(\xi_0,\xi_1,\dots)$ we have $[\Phi(\pa_{u_0},\pa_{u_1},\dots),u_j]=(\pa_{\xi_j}\Phi)(\pa_{u_0},\pa_{u_1},\dots)$.
In particular when $\Phi(\xi_0,\xi_1,\cdots) = \exp(\pm\frac 12\sum_{i,j\geq 0}\nu_{i,j}\xi_i\xi_j)$, with 
\be
\nu_{i,j} \defis (-1)^{\frac{i-j}2}\frac{B_{i+j+2}}{i+j+2},
\ee we get
\be
\label{comm}
[\B^{\pm 1},u_j] \= \mp\sum_{i\geq 0}\nu_{i,j}\frac{\pa}{\pa u_i}\B^{\pm 1}.
\ee

We claim that $[\B,(k+2+\mathcal D)\pa_x]=0$. 
Indeed, by~\eqref{comm} we have
\be
[\B,\pa_x] \= -\sum_{i,j\geq 0}\nu_{i,j+1}\frac{\pa^2}{\pa u_i\pa u_j}\B =0,
\ee
because $\nu_{i,j+1}=-\nu_{j,i+1}\hspace{1pt},$ and $[\B,\mathcal D] =0$.
Therefore,
\be
(k+2+\mathcal D)\pa_x\wt g_{k+1}
=\B^{-1}(k+2+\mathcal D)\pa_xg_{k+1}
=\B^{-1}\mathcal Rg_k
=\B^{-1}\mathcal R\mathcal \B \wt g_k\,
\ee
The proof is complete once we show the identity $\B^{-1}\mathcal R\B=\mathcal R_1\hspace{1pt},$ or, equivalently,
\be
\label{tobeproved}
[\B,\mathcal R_1]\=\mathcal R_2\mathcal \B.
\ee
The operator $\mathcal R_1$ consists of three parts, namely, $\mathcal R_1=\mathcal R_1^{(a)}+\mathcal R_1^{(b)}+\mathcal R_1^{(c)}$ with
\begin{align}
\mathcal R_1^{(a)} &\= 
\sum_{\ell\geq 0}(\epsilon-\mu)\sum_{g\geq 1}(\epsilon\mu)^{g-1}\frac{|B_{2g}|}{(2g)!} u_{\ell+2g+1}\frac{\pa}{\pa u_\ell},
\\
\mathcal R_1^{(b)} &\= \frac 12 \sum_{\ell\geq 0}\sum_{k=0}^{\ell+1}\binom{\ell+1}{k}u_ku_{\ell+1-k}\frac{\pa}{\pa u_\ell},
\\
\mathcal R_1^{(c)} &\=-\frac 12 \sum_{\ell,m\geq 0}\frac{(\ell+1)!(m+1)!}{(\ell+m+3)!}u_{\ell+m+3}\frac{\pa^2}{\pa u_\ell\pa u_m}.
\end{align}
We compute separately each contribution, using~\eqref{comm}.
\begin{enumerate}[{\hspace{-18pt}\upshape (a)}]
\item
We have
\begin{align}
\nonumber
[\B,\mathcal R_1^{(a)}] &\= (\epsilon-\mu) \sum_{\ell\geq 0}\sum_{g\geq 1}(\epsilon\mu)^{g-1}\frac{|B_{2g}|}{(2g)!}[\B,u_{\ell+2g+1}]\frac{\pa}{\pa u_\ell} 
\\
&\=(\epsilon-\mu) \sum_{i,\ell\geq 0}\sum_{g\geq 1}(\epsilon\mu)^{g-1}\frac{|B_{2g}|}{(2g)!}\nu_{i,\ell+2g+1}\frac{\pa^2}{\pa u_i\pa u_\ell}\mathcal \B \= 0,
\end{align}
because $\nu_{i,\ell+k}=-\nu_{\ell,i+k}$ for any odd $k$.
\item
We first compute
\begin{align}
\nonumber
[\B,u_{j_1}u_{j_2}] &\= u_{j_1}[\B,u_{j_2}] \+ [\B,u_{j_1}]u_{j_2}
\\
\nonumber
&\=-\sum_{i\geq 0}\Bigl(\nu_{i,j_2}u_{j_1}\frac{\pa}{\pa u_i}\B +\nu_{i,j_1}\frac{\pa}{\pa u_i}\B u_{j_2}\Bigl)
\\
\nonumber
&\=-\sum_{i\geq 0}\Bigl(\nu_{i,j_2}u_{j_1}\frac{\pa}{\pa u_i}+\nu_{i,j_1}\frac{\pa}{\pa u_i}u_{j_2}\Bigr)\B \meno \sum_{i\geq 0}\nu_{i,j_1}\frac{\pa}{\pa u_i}[\B,u_{j_2}]
\\
&\=-\bigg(\sum_{i\geq 0}\Bigl(\nu_{i,j_2}u_{j_1}\frac{\pa}{\pa u_i}+\nu_{i,j_1}u_{j_2}\frac{\pa}{\pa u_i}\Bigr)
+\nu_{j_1,j_2}
-\sum_{i_1,i_2\geq 0}\nu_{i_1,j_1}\nu_{i_2,j_2}\frac{\pa^2}{\pa u_{i_1}\pa u_{i_2}}\bigg)\B.
\end{align}
Therefore,
\begin{align}
\nonumber
[\B,\mathcal R_1^{(b)}]\B^{-1} \= & \frac 12 \sum_{\ell\geq 0}\sum_{k=0}^{\ell+1}\binom{\ell+1}{k}[\B,u_ku_{\ell+1-k}]\B^{-1} \frac{\pa}{\pa u_\ell}
\\
\nonumber
\=& -\frac 12\sum_{i,\ell\geq 0}\sum_{k=0}^{\ell+1}\binom{\ell+1}{k}\bigg(\nu_{i,\ell+1-k}u_{k}+\nu_{i,k}u_{\ell+1-k}\bigg)\frac{\pa^2}{\pa u_i\pa u_\ell}
\+
\\
\nonumber
&
-\frac 12\sum_{\ell\geq 0}\sum_{k=0}^{\ell+1}\binom{\ell+1}{k}\nu_{k,\ell+1-k}\frac{\pa}{\pa u_\ell} \+
\\
\label{vanishes}
& +\frac 12\sum_{i_1,i_2,\ell\geq 0}\sum_{k=0}^{\ell+1}\binom{\ell+1}{k}\nu_{i_1,k}\nu_{i_2,\ell+1-k}\frac{\pa^3}{\pa u_{i_1}\pa u_{i_2}\pa u_\ell}.
\end{align}
Since
\be
\sum_{k=0}^{\ell+1}\binom{\ell+1}{k}\nu_{k,\ell+1-k} \= (-1)^{\frac{\ell+1}2}\frac{B_{\ell+3}}{\ell+3}\sum_{k=0}^{\ell+1}\binom{\ell+1}{k}(-1)^k \= 0, \qquad \ell\geq 0,
\ee
the second sum in~\eqref{vanishes} vanishes for all $\ell\geq 0$.
\item Finally, we have
\be
[\B,\mathcal R_1^{(c)}] \=\frac 12\sum_{i,\ell,m\geq 0}\frac{(\ell+1)!(m+1)!}{(\ell+m+3)!}\nu_{i,\ell+m+3}\frac{\pa^3}{\pa u_i\,\pa u_\ell\,\pa u_m}\B.
\ee
\end{enumerate}
Combining these three computations we obtain, denoting for convenience $b_k:=B_k/k$,
\begin{align}
\nonumber
[\B,\mathcal R_1]\B^{-1} \= &-\frac 12\sum_{i,\ell\geq 0}\sum_{k=0}^{\ell+1}\binom{\ell+1}{k}\bigg(
(-1)^{\frac{i+k-\ell-1}2}b_{i+k+\ell+3}u_{k}+(-1)^{\frac{i-k}2}b_{i+k+2}u_{\ell+1-k}
\bigg)\frac{\pa^2}{\pa u_i\pa u_\ell} \+
\\
\nonumber
&
+\frac 12\sum_{i_1,i_2,\ell\geq 0}\sum_{k=0}^{\ell+1}\binom{\ell+1}{k}(-1)^{\frac{i_1+i_2-\ell-1}2} b_{i_1+k+2}b_{i_2+\ell+3-k}\frac{\pa^3}{\pa u_{i_1}\pa u_{i_2}\pa u_\ell} \+
\\
\label{last}
&
+\frac 12\sum_{i,\ell,m\geq 0}\frac{(\ell+1)!(m+1)!}{(\ell+m+3)!}(-1)^{\frac{\ell+m+3-i}2}b_{i+\ell+m+5}\frac{\pa^3}{\pa u_i\pa u_\ell\pa u_m}.
\end{align}
It is easy to check that the first line on the right-hand side of~\eqref{last} equals $\mathcal R_2$ (see~\eqref{eq:RILW2}).
To complete the proof we need to show that the last two lines in~\eqref{last} cancel each other.
To this end, we first recall the following theorem by Skoruppa.
\begin{theorem}[\cite{Skoruppa}]
\label{thm:Skoruppa}
Suppose that the symmetric homogeneous bivariate polynomial $H(x,y) = \sum_{\nu=0}^n h_\nu x^\nu y^{n-\nu}$ of even positive degree~$n$ satisfies $H(x,y)=H(y,y-x)$.
Then
\be
\sum_{\substack{0<\nu<n \\ \nu \text{\em{ odd}}}}h_\nu \,\mathbb{G}_{\nu+1}\, \mathbb{G}_{n+1-\nu} \= h\, \mathbb{G}_{n+2}\meno\frac {h_1}n\,q\frac{\d}{\d q}\mathbb{G}_{n},
\ee
where $h:=-\frac 12\int_0^1 H(1,y)\,\d y$ and $\mathbb{G}_n$ is the Eisenstein series~\eqref{eq:Eisenstein}.
\end{theorem}
Since $\mathbb{G}_n$ has a Fourier series expansion in the upper half plane whose constant term is $-b_{n}/2$ we get
\be
\label{eq:Bernoulliconsequence}
\sum_{\substack{0<\nu<n \\ \nu \text{\em{ odd}}}}h_\nu \,b_{\nu+1}\, b_{n+1-\nu} \= -2\,h\, b_{n+2},
\ee
with the notations of Theorem~\ref{thm:Skoruppa}.
We need the following specialization: given positive integers $a_1,a_2,a_3$ with $a_1+a_2+a_3$ even, the polynomial
\begin{align}
\nonumber
H(x,y)\defis& \sum_{\pi\in S_3}x^{a_{\pi(1)}}(-y)^{a_{\pi(2)}}(y-x)^{a_{\pi(3)}}  \\
\=& \sum_{\pi\in S_3}\sum_{k=0}^{a_{\pi(3)}}\binom{a_{\pi(3)}}{k}(-1)^{a_{\pi(2)}+k}x^{a_{\pi(1)}+k}y^{a_{\pi(2)}+a_{\pi(3)}-k}
\end{align}
satisfies the condition of Theorem~\ref{thm:Skoruppa}, because
\begin{align}
\nonumber
H(y,y-x) &\= \sum_{\pi\in S_3}y^{a_{\pi(1)}}(x-y)^{a_{\pi(2)}}(-x)^{a_{\pi(3)}}
\\
&\=(-1)^{a_1+a_2+a_3}\sum_{\pi\in S_3}(-y)^{a_{\pi(1)}}(y-x)^{a_{\pi(2)}}(x)^{a_{\pi(3)}}\=H(x,y). 
\end{align}
Hence, by~\eqref{eq:Bernoulliconsequence}
\begin{multline}
\sum_{\pi\in S_3}\sum_{k=0}^{a_{\pi(3)}}\binom{a_{\pi(3)}}{k}(-1)^{a_{\pi(2)}+k}b_{a_{\pi(1)}+k+1}\,b_{a_{\pi(2)}+a_{\pi(3)}-k+1}
\\ \= 
b_{a_1+a_2+a_3+2}\sum_{\pi\in S_3}(-1)^{a_{\pi(1)}}\frac{a_{\pi(2)}!\,a_{\pi(3)}!}{(a_{\pi(2)}+a_{\pi(3)}+1)!}.
\end{multline}
Note that $a_{\pi(2)}+a_{\pi(3)}-k+1$ is even in the left-hand side of the last identity and so, multiplying both sides by $(-1)^{\frac{a_1+a_2+a_3}2}$, we may write it as
\begin{multline}
\sum_{\pi\in S_3}\sum_{k=0}^{a_{\pi(3)}}\binom{a_{\pi(3)}}{k}(-1)^{\frac{a_{\pi(1)}+a_{\pi(2)}-a_{\pi(3)}}2}b_{a_{\pi(1)}+k+1}\,b_{a_{\pi(2)}+a_{\pi(3)}-k+1}
\\=-b_{a_1+a_2+a_3+2}\sum_{\pi\in S_3}(-1)^{\frac{a_{\pi(2)}+a_{\pi(3)}-a_{\pi(1)}}2}\frac{a_{\pi(2)}!\,a_{\pi(3)}!}{(a_{\pi(2)}+a_{\pi(3)}+1)!}.
\end{multline}
It is clear by this identity that the two cubic operators in the $\pa/\pa u_i$'s, which appear in the second and third line of~\eqref{last}, cancel each other.
\end{proof}

\begin{corollary}
\label{corollary:homogeneous}
For all $k\geq -2$, the reduced density $\wt g_k^{\sf ILW}(\u;\epsilon)$ is homogeneous of weight $k+2$ if we assign weight $i+1$ to $u_i\hspace{1pt},$ $+1$ to $c$, and $-1$ to $\epsilon$ and $\mu$.
\end{corollary}

\begin{proof}
Let $\mathcal W \defis \sum_{i\geq 0}(i+1)u_i\pdv{}{u_i}-\epsilon\pdv{}{\epsilon}-\mu\pdv{}{\mu}$ be the grading operator with respect to the weights of the statement.
We need to prove that $\mathcal W\,\wt g_k^{\sf ILW} \= (k+2)\wt g_k^{\sf ILW}$ for $k\geq -1$ (the case $k=-2$ being trivial).
It is straightforward to verify the commutation relations
\be
\biggl[\mathcal W,\pdv{}{u_0}\biggr] \= -\pdv{}{u_0},\qquad
[\mathcal W,(k+2+\mathcal D)\pa_x] \= (k+2+\mathcal D)\pa_x,\qquad
[\mathcal W,\mathcal R_1^{\sf ILW}]\=2\mathcal R_1^{\sf ILW},
\ee
where $\mathcal D$ is given in~\eqref{eq:DILW}.
It follows from these relations and~\eqref{eq:recursionreduced1}-\eqref{eq:recursionreduced2} that the polynomials $f_k:=\mathcal W\,\wt g_k^{\sf ILW}$ satisfy the recursion
\begin{align}
\pdv{}{u_0}f_{k+1}\=&f_k+\wt g_k^{\sf ILW},\\
(k+2+\mathcal D)\pa_xf_{k+1}\=&\mathcal R_1^{\sf ILW}\,(f_k+\wt g_k^{\sf ILW}),\qquad k\geq -1,
\end{align}
with initial condition $f_{-1}=u_0\hspace{1pt}.$
This recursion uniquely determines (by an argument parallel to that in Remark~\ref{remark:recursionwellposed}) all $f_k$'s, for $k\geq -1$.
On the other hand, this recursion is satisfied by $f_k \= (k+2)\wt g_k\hspace{1pt},$ and the proof is complete.
\end{proof}

Let $\wt M[c,\epsilon,\mu]:=\wt M\otimes\Q[c,\epsilon,\mu] =: \bigoplus_k\wt M[c,\epsilon,\mu]_{k+2}\hspace{1pt},$ graded by the quasimodular weight and by assigning weight $+1$ to $c$ and $-1$ to $\epsilon$ and $\mu$.
The central result of this section now follows directly from Theorem~\ref{thm:main} and Corollary~\ref{corollary:homogeneous}.

\begin{framed}
\begin{theorem}
\label{thm:ILW}
For all $k\geq -2$, we have
\be
\opq {G^{\sf ILW}_k(\epsilon,\mu)} \,\in\,\wt M[c,\epsilon,\mu]_{k+2}\hspace{1pt}.
\ee
\end{theorem}
\end{framed}

The quantum KdV hierarchy mentioned in Section~\ref{sec:intro} corresponds to the special case $\mu=0$ of the ILW hierarchy, namely
\be
g_k^{\sf KdV}(\u;\epsilon) \= \left.g_k^{\sf ILW}(\u;\epsilon,\mu)\right|_{\mu=0},\qquad
G_k^{\sf KdV}(\epsilon) \= \left.G_k^{\sf ILW}(\epsilon,\mu)\right|_{\mu=0}.
\ee
Therefore, Theorem~\ref{thm:KdV} is a direct corollary of Theorem~\ref{thm:ILW}.

\section*{Acknowledgments}
We are indebted to Boris Dubrovin, Di Yang and Don Zagier for valuable conversations and suggestions.
We extend our thanks to the anonymous referee for useful comments and suggestions for improving the presentation.
GR acknowledges support from the Fonds de la Recherche Scientifique-FNRS under EOS project O013018F and from the FCT Grant 2022.07810.CEECIND.

\appendix

\section{Closed formulas for the quantum KdV hierarchy for \texorpdfstring{$\epsilon\to 0$ and ${\epsilon \to\infty}$}{ε to 0 and infinity}}\label{app}

\subsection{Dispersionless limit \texorpdfstring{$\epsilon\to 0$}{ε to 0}.} 
An explicit generating function for the Hamiltonian densities in the case $\epsilon = 0$ is due to Eliashberg, see \cite{Eliashberg} and \cite[Proposition~4.1]{BuryakRossi2}.
Namely, it is known that
\be
\label{eq:Eliashberg}
\sum_{k\geq -2}y^{k+2}\,g_{k}^{\sf KdV}(\u;\epsilon=0) \= \frac {\exp(y\, S(\i y \partial_x)u_0) }{S(y)}, \qquad S(y) \defis \frac{\sinh(y/2)}{y/2} \= \sum_{k\geq 0}\frac{y^{2k}}{4^k(2k+1)!},
\ee
where $\partial_x$ is defined by~\eqref{eq:D}. We observe here the simplification of this formula when we consider the reduced densities instead.
\begin{proposition}
\label{prop:Hopf}
Denoting $\wt g_k^{\sf KdV}:=\mathcal B^{-1}g_k^{\sf KdV}$, we have
\be
\label{eq:simplifiedEliashberg}
\sum_{k\geq -2}y^{k+2}\,\wt g_{k}^{\sf KdV}(\u;\epsilon=0) \=
\exp(y\, S(\i y \partial_x)u_0).
\ee
\end{proposition}
\begin{remark}
By assigning degree $-1$ to $y$ we see directly that the right-hand side in~\eqref{eq:simplifiedEliashberg} is of homogeneous weight zero (provided $u_k$ has weight $k+1$), hence we recover that $\wt g_{k}^{\sf KdV}(\u;\epsilon=0)$ has degree $k+2$.
Actually, since by the definition of the operator $\mathcal B$ we have $g_{k}^{\sf KdV}=\wt g_{k}^{\sf KdV}\+$lower degree terms, it is not difficult to derive~\eqref{eq:simplifiedEliashberg} from~\eqref{eq:Eliashberg}.
\end{remark}

\subsection{Limit \texorpdfstring{$\epsilon\to \infty$}{ε to infinity}.} 
We also provide formulas for the (sub)leading terms in $g_k^{\sf KdV}(\u;\epsilon)$ as $\epsilon\to\infty$. 

\begin{proposition}
\label{prop:infinitedispersion}
For all $k\geq -1$, the density~$g_k^{\sf KdV}(\u;\epsilon)$ is a polynomial in $\epsilon$ of degree $k+1$, whose leading and subleading terms are given by
\begin{align}
[\epsilon^{k+1}]\,g_k^{\sf KdV}(\u;\epsilon) \, &=\, \frac{u_{2k+2}}{24^{k+1}(k+1)!},& & k\geq -1,
\\
[\epsilon^k]\,g_k^{\sf KdV}(\u;\epsilon) \, &=\, -\frac 1{(-4)^k(2k+1)!!}\frac{B_{2k+2}}{4k+4}\+\frac 1{24^k}\sum_{j=0}^{2k}d(j,k) u_ju_{2k-j}\hspace{1pt},& & k\geq 0,
\end{align}
where the coefficients~$d(j,k)$ {\upshape(}for $j,k>0${\upshape)} are given by the generating series
\be
\sum_{j,k\geq 0}(2k+1)!!\,d(j,k)\,y^jx^k \=  \frac 1{2\sqrt{1-2x(1+y)^2}\left(1-2x(1-y+y^2)\right)}
\ee
and in particular satisfy $d(j,k)=0$ for $j>2k$ and $d(j,k)=d(2k-j,k)$ for $j,k\geq 0$ with $j\leq 2k$.
\end{proposition}
\begin{proof}
It is straightforward from the definition of $\mathcal B$ and the identity
\be
\sum_{j=0}^{2k}(-1)^j\,d(j,k) \= \frac 1{(2k+1)!!}\,[x^k] \frac 1{2(1-6x)} \= \frac{6^k}{2(2k+1)!!}
\ee
to check that the statement is equivalent to the fact that the reduced density~$\wt g_k^{\sf KdV}:=\mathcal B^{-1}g_k^{\sf KdV}$ is a polynomial in~$\epsilon$ of degree~${k+1}$ for $k\geq -1$ whose leading and subleading terms are given by
\begin{align}
\label{toprove1}
[\epsilon^{k+1}]\,\wt g_k^{\sf KdV}(\u;\epsilon) \, &=\, \frac{u_{2k+2}}{24^{k+1}(k+1)!},& & k\geq -1,
\\
\label{toprove2}
[\epsilon^k]\,\wt g_k^{\sf KdV}(\u;\epsilon) \, &=\, \frac 1{24^k}\sum_{j=0}^{2k}d(j,k) u_ju_{2k-j}\hspace{1pt},& & k\geq 0.\end{align}
Therefore, it suffices to show~\eqref{toprove1} and~\eqref{toprove2}.
The~$\wt g_k$'s are determined by~\eqref{eq:recursionreduced1} and~\eqref{eq:recursionreduced2} with $\mu=0$, namely,
\be
\label{eq:recKdV}
\Bigl(k+2+\epsilon\pdv{}{\epsilon}\Bigr)\pa_x\wt g_{k+1}^{\sf KdV} \= (\mathcal R_0+\epsilon\,\mathcal R_1)\wt g_k^{\sf KdV},\qquad
\pdv{\wt g_{k+1}^{\sf KdV}}{u_0}\= \wt g_k^{\sf KdV},\qquad\quad k\geq -1,
\ee
with initial condition $\wt g_{-1}=u_0\hspace{1pt},$ where
\begin{align}
\label{eq:R0KdV}
\mathcal R_0\defis &\frac 12\sum_{i,j\geq 0}\left(\frac{(i+1)!}{(i+1-j)!j!}u_{i+1-j}u_j\pdv{}{u_i}-\frac{(i+1)!(j+1)!}{(i+j+3)!}u_{i+j+3}\pdv{^2}{u_i\pa u_j}\right),
\\
\mathcal R_1\defis & \frac 1{12}\sum_{i\geq 0}u_{i+3}\pdv{}{u_i}.
\end{align}
It follows that $\wt g_k$ is a polynomial of degree at most~${k+1}$ in~$\epsilon$.
Moreover, the leading term satisfies the recursion
\be
\left(2k+4\right)\pa_x\bigl([\epsilon^{k+2}]\wt g_{k+1}^{\sf KdV}\bigr) \= \mathcal R_1\bigl([\epsilon^{k+1}]\wt g_k^{\sf KdV}\bigr),\qquad\qquad k\geq -1,
\ee
with $[\epsilon^0]\,\wt g_{-1}^{\sf KdV} \= u_0\hspace{1pt}.$
Hence, for $k\geq -1$ it follows that $[\epsilon^{k+1}]\,\wt g_k^{\sf KdV}(\u;\epsilon) = \frac{u_{2k+2}}{24^{k+1}(k+1)!}\+c_k$ for some constants~$c_k$ depending on~$k$ only.
By Corollary~\ref{corollary:homogeneous}, $[\epsilon^{k+1}]\,\wt g_k^{\sf KdV}(\u;\epsilon)$ must be of homogeneous weight~${2k+3}$, hence $c_k=0$.
Then the subleading term is determined by the recursion
\be
\label{eq:subleadingrec}
\left(2k+3\right)\pa_x\bigl([\epsilon^{k+1}]\,\wt g_{k+1}^{\sf KdV}\bigr) \= \frac{\mathcal R_0(u_{2k+2})}{24^{k+1}(k+1)!}+\mathcal R_1\bigl([\epsilon^{k}]\,\wt g_k^{\sf KdV}\bigr),\qquad k\geq 0
\ee
with $[\epsilon^0]\wt g_0^{\sf KdV}\=u_0^2/2.$ Therefore, $[\epsilon^k]\,\wt g_k^{\sf KdV}=24^{-k}\sum_{j=0}^{2k}d(j,k)u_ju_{2k-j}$ for some coefficients~$d(j,k)$.
This is in principle only true up to a constant depending on~$k$ only, however Corollary~\ref{corollary:homogeneous} again implies that this constant vanishes.
Here the coefficients are assumed to satisfy $d(j,k)=d(2k-j,k)$ and, as a consequence of~\eqref{eq:subleadingrec}, are subject to the recurrence
\be
(2k+3)(d(j,k+1)+d(j-1,k+1)) \= \frac 1{2(k+1)!}\binom{2k+3}j\+2(d(j-3,k)+d(j,k)),\quad j,k\geq 0,
\ee
where $d(j,k)=0$ for $j<0$ or $j>2k$, with initial condition $d(0,0)\=1/2$.
Therefore
\be
(2k+3)(1+y)\,\Delta_{k+1}(y) \= \frac{(1+y)^{2k+3}}{2(k+1)!}+2(1+y^3)\,\Delta_k(y),\qquad \Delta_k(y)\defis\sum_{j=0}^{2k}d(j,k)\,y^j.
\ee
Dividing by $(1+y)$, multiplying by $x^{k+1}(2k+1)!!$, and summing over $k\geq 0$ we obtain
\be
D(x,y)-\frac 12\=\frac{1}{2} \biggl(\frac{1}{\sqrt{1-2 x (y+1)^2}}-1\biggr)\+2x(1-y+y^2)D(x,y),
\ee
where $D(x,y) := \sum_{k\geq 0}(2k+1)!!\,\Delta_k(y)\,x^k$, and~\eqref{toprove2} follows.
\end{proof}

In particular, we obtain the following immediate consequence, whose proof is omitted.

\begin{corollary}
\label{corollary:infinitedispersion}
For $k\geq 0$, $G_k^{\sf KdV}(\epsilon)=\op{g_k^{\sf KdV}(\u;\epsilon)}$ is a polynomial in $\epsilon$ of degree~$k$ whose leading coefficient is
\be
[\epsilon^k]\,G_k^{\sf KdV}(\epsilon) \= \frac{c^2}2\delta_{k,0}\+\frac 1{(-4)^k(2k+1)!!}\biggl(-\frac{B_{2k+2}}{4k+4}\+\sum_{j\geq 1}j^{2k+1}p_j\pdv{}{p_j}\biggr).
\ee
\end{corollary}

\section{Tables of quasimodular forms}\label{app2}

To illustrate our main theorems, we provide the following examples. As above, $\mathbb G_k$ is the Eisenstein series~\eqref{eq:Eisenstein}. Note that $\mathbb{G}_8=120\mathbb{G}_4^2$.

\subsection{}
Table of $q$-series associated with the first few quantum KdV Hamiltonian operators.
\[
\def\arraystretch{1.5}
\allowdisplaybreaks
\begin{array}{l||l}
k &  \bigl\lbrace G_{k}^{\sf KdV}(\epsilon)\bigr\rbrace_{\! q}  \\\hline\hline
-2 &  1
\\\hline
-1 & c
\\\hline
0 & \mathbb{G}_2+\frac{c^2}{2}
\\\hline
1 & c \,\mathbb{G}_2 +  \frac{c^3}{6} \meno \frac{\epsilon}{24}( 2 \mathbb{G}_4)
\\\hline
2 & \frac{1}{2}\mathbb{G}_2^2 + \frac{1}{12}\mathbb{G}_4 + \frac{c^2}{2} \mathbb{G}_2 + \frac{c^4}{24} \meno \frac{\epsilon}{24}(2c\, \mathbb{G}_4) \+ \bigl(\frac{\epsilon}{24}\bigl)^{\! 2}\frac{12}{5} \mathbb{G}_6 
\\\hline
3 & \frac c2 \,\mathbb G_2^2 +\frac c{12}\mathbb G_4+\frac {c^3}6\mathbb G_2+\frac {c^5}{120}
\meno\frac{\epsilon}{24}(\frac 16\mathbb G_6+2\mathbb G_4\mathbb G_2+c^2\mathbb G_4)
\+\bigl(\frac{\epsilon}{24}\bigl)^{\! 2}(\frac {12}5c\,\mathbb G_6)
\meno\bigl(\frac{\epsilon}{24}\bigl)^{\! 3}(\frac {72}{35}\mathbb G_8)
\end{array}
\]

\subsection{}
Table of non-zero $q$-series of the form $\left\lbrace\op{\B g} \right\rbrace_{\! q}$ for monomials $g\in \Q[\u]$ of weight up to $9$ (recalling that $u_k$ has weight $k+1$).

\smallskip

\resizebox{.95\textwidth}{!} 
{$\def\arraystretch{1.5}
\allowdisplaybreaks
\begin{array}{ll||ll}
g & \left\lbrace\op{\B g} \right\rbrace_{\! q}  & 
	g & \left\lbrace\op{\B g} \right\rbrace_{\! q}  \\\hline\hline
u_0^2 & 2\mathbb{G}_2+c^2 
	&u_0^5u_2 &  - 120\mathbb{G}_4\mathbb{G}_2^2 - 120\mathbb{G}_4\mathbb{G}_2\,c^2-10\mathbb{G}_4\,c^4\\\hline
u_0^4 &  12 \mathbb{G}_2^2 + 12 \mathbb{G}_2\,c^2 + c^4 
	&u_0^4u_1^2& 24\mathbb{G}_4\mathbb{G}_2^2  + 24\mathbb{G}_4\mathbb{G}_2\,c^2+ 2\mathbb{G}_4\,c^4\\\hline
u_0u_2 & -2\mathbb{G}_4 
	& u_0^3u_4& 12\mathbb{G}_6\mathbb{G}_2 + 6\mathbb{G}_6\,c^2 \\\hline
u_1^2 & 2\mathbb{G}_4 
	& u_0^2u_1u_3 &  - 4\mathbb{G}_6\mathbb{G}_2-2\mathbb{G}_6\,c^2\\\hline
u_0^6 &   120\mathbb{G}_2^3 + 180\mathbb{G}_2^2\,c^2+ 30\mathbb{G}_2\,c^4 + c^6 
	& u_0^2u_2^2 & 8\mathbb{G}_4^2 + 4\mathbb{G}_6\mathbb{G}_2 + 2\mathbb{G}_6\,c^2 \\\hline
u_0^3 u_2 &  - 12\mathbb{G}_4\mathbb{G}_2-6\mathbb{G}_4\,c^2 	
	& u_0u_1^2u_2 & -4\mathbb{G}_4^2 \\\hline
u_0^2 u_1^2 & 4\mathbb{G}_4\mathbb{G}_2 + 2\mathbb{G}_4\,c^2
	& u_0u_6 & -2\mathbb{G}_8 \\\hline
u_0 u_4 & 2\mathbb{G}_6
	& u_1u_5 &  2\mathbb{G}_8 \\\hline
u_1 u_3 & -2\mathbb{G}_6
	& u_1^4 & 12\mathbb{G}_4^2 \\\hline
u_2^2 & 2\mathbb{G}_6 
	& u_2u_4 &  -2\mathbb{G}_8\\\hline
u_0^8 & 1680\mathbb{G}_2^4 + 3360\mathbb{G}_2^3\,c^2 + 840\mathbb{G}_2^2\,c^4 + 56\mathbb{G}_2\,c^6 + c^8 
	& u_3^2 &  2\mathbb{G}_8
\end{array}$
}


\begin{thebibliography}{10}

\bibitem{BlochOkounkov}
S.~Bloch \& A.~Okounkov.
\newblock ``The character of the infinite wedge representation''.
\newblock {\em Adv. Math.} 149 (2000), no. 1, 1--60. 

\bibitem{BSTV2014}
G.~Bonelli, A.~Sciarappa, A.~Tanzini, \& P.~Vasko.
\newblock ``Six-dimensional supersymmetric gauge theories, quantum cohomology of instanton moduli spaces and gl(N) Quantum Intermediate Long Wave Hydrodynamics''.
\newblock {\em J. High Energ. Phys.} 2014 (2014), no. 141, 29 pp.

\bibitem{Buryak}
A.~Buryak.
\newblock ``Double ramification cycles and integrable hierarchies''.
\newblock {\em Comm. Math. Phys.} 336 (2015), no. 3, 1085--1107. 

\bibitem{BuryakDubrovinGuereRossi}
A.~Buryak, B.~Dubrovin, J.~Gu\'er\'e, \& P.~Rossi.
\newblock ``Integrable systems of double ramification type''.
\newblock {\em Int. Math. Res. Not. IMRN} 2020, no. 24, 10381--10446.

\bibitem{BuryakRossi1}
A.~Buryak \& P.~Rossi.
\newblock ``Recursion relations for double ramification hierarchies''.
\newblock {\em Comm. Math. Phys.} 342 (2016), no. 2, 533--568.

\bibitem{BuryakRossi2}
A.~Buryak \& P.~Rossi.
\newblock ``Double ramification cycles and quantum integrable systems''.
\newblock {\em Lett. Math. Phys.} 106 (2016), no. 3, 289--317.

\bibitem{ChenMollerZagier}
D.~Chen, M.~M\"{o}ller, \& D.~Zagier, ``Quasimodularity and large genus limits
  of {S}iegel--{V}eech constants,'' {\em J. Amer. Math. Soc.} 31 (2018), no.~4,
  1059--1163.

\bibitem{ChenMollerSauvagetZagier}
D.~Chen, M.~M\"{o}ller, A.~Sauvaget, \& D.~Zagier, ``Masur--{V}eech volumes
  and intersection theory on moduli spaces of {A}belian differentials,'' {\em
  Invent. Math.} 222 (2020), no.~1, 283--373.
  
\bibitem{Dijkgraaf}
R.~Dijkgraaf, ``Mirror symmetry and elliptic curves,'' {\em The moduli space of curves ({T}exel {I}sland, 1994)}, 149--163,
{\em Progr. Math.}, 129, Birkhäuser Boston, Boston, MA, 1995.

\bibitem{Dubrovin}
B.~Dubrovin.
\newblock ``Symplectic field theory of a disk, quantum integrable systems, and Schur polynomials''.
\newblock {\em Ann. Henri Poincar\'e} 17 (2016), no. 7, 1595--1613. 

\bibitem{Eliashberg}
Y.~Eliashberg.
\newblock ``Symplectic field theory and its applications''.
\newblock {\em  International Congress of Mathematicians. Vol. I}, 217--246, {\em Eur. Math. Soc., Zürich}, 2007.

\bibitem{EliashbergGiventalHofer}
Y.~Eliashberg, A.~Givental, \& H.~Hofer.
\newblock ``Introduction to symplectic field theory''.
\newblock {\em Geom. Funct. Anal.} 2000, Special Volume, Part II, 560--673.

\bibitem{Engel}
P.~Engel, ``Hurwitz theory of elliptic orbifolds, {I},'' {\em Geom. Topol.} 25 (2021), no.~1, 229--274.

\bibitem{EskinOkounkov}
A.~Eskin \& A.~Okounkov, ``Asymptotics of numbers of branched coverings of a
  torus and volumes of moduli spaces of holomorphic differentials,'' {\em
  Invent. Math.} 145 (2001), no.~1, 59--103.

\bibitem{Goulden}
I.~P.~Goulden.
\newblock ``A differential operator for symmetric functions and the combinatorics of multiplying transpositions''.
\newblock {\rm Trans. Amer. Math. Soc.} 344 (1994), no. 1, 421--440. 

\bibitem{HahnIttersumLeid}
M.~A.~Hahn, J.-W.~M.~van Ittersum, \& F.~Leid.
\newblock ``Triply mixed coverings of arbitrary base curves: quasimodularity, quantum curves and a mysterious topological recursion''.
\newblock {\it Ann. Inst. Henri Poincar\'e D}~9 (2022), no.~2, 239--296.

\bibitem{vanIttersum1}
J.-W.~M.~van Ittersum.
\newblock ``When is the Bloch--Okounkov $q$-bracket modular?''.
\newblock {\em Ramanujan J.} 52 (2020), no. 3, 669--682. 

\bibitem{vanIttersumSymmetric}
J.-W.~M.~van Ittersum.
\newblock ``A symmetric Bloch--Okounkov theorem''.
\newblock {\em Res. Math. Sci.} 8 (2021), no. 2, Paper No. 19, 42 pp. 

\bibitem{IvanovOlshanski}
V.~Ivanov \& G.~Olshanski, 
\newblock ``Kerov's central limit theorem for the
  {P}lancherel measure on {Y}oung diagrams,''
\newblock {\em Symmetric functions 2001: surveys of developments and perspectives}, 93--151, vol.~74 of {\em NATO Sci. Ser. II
  Math. Phys. Chem.}, Kluwer Acad. Publ., Dordrecht, 2002.

\bibitem{KerovOlshanski}
S.~{Kerov} \& G.~{Olshanski}, ``{Polynomial functions on the set of Young
  diagrams.},'' {\em {C. R. Acad. Sci., Paris, S\'er. I}}, 319 (1994), no.~2,
  121--126.

\bibitem{KontsevichManin}
M.~Kontsevich \& Yu.~Manin.
\newblock ``Gromov--Witten classes, quantum cohomology, and enumerative geometry''.
\newblock {\em Comm. Math. Phys.} 164 (1994), no. 3, 525--562. 

\bibitem{Macdonald}
I.~G.~Macdonald.
\newblock ``Symmetric functions and Hall polynomials''.
\newblock Oxford Classic Texts in the Physical Sciences. {\em The Clarendon Press, Oxford University Press, New York}, 2015.

\bibitem{MiwaJimboDate}
T.~Miwa, M.~Jimbo, \& E.~Date.
\newblock ``Solitons: Differential Equations, Symmetries and Infinite Dimensional Algebras''.
\newblock Cambridge Tracts in Mathematics, 135. {\em Cambridge University Press, Cambridge,} 2000.

\bibitem{Ochiai}
H.~Ochiai, ``Counting functions for branched covers of elliptic curves and
  quasi-modular forms,'' 
\newblock Representation theory of vertex operator algebras and related topics
  (Japanese) (Kyoto, 2000), No.~1218 (2001), 153--167, 2001.
  
\bibitem{OkounkovPandharipande}
A.~Okounkov \& R.~Pandharipande, ``Gromov-{W}itten theory, {H}urwitz theory,
  and completed cycles,'' {\em Ann. of Math. (2)}, 163 (2006), no.~2,
  517--560.
  
\bibitem{Rios}
R.~R\'ios-Zertuche.
\newblock ``An introduction to the half-infinite wedge.''
\newblock {\em Mexican mathematicians abroad: recent contributions}, 197–237,
Contemp. Math., 657, Aportaciones Mat., {\em Amer. Math. Soc., Providence, RI}, 2016.

\bibitem{Rossi}
P.~Rossi.
\newblock ``Gromov-Witten invariants of target curves via symplectic field theory''.
\newblock {\em J. Geom. Phys.} 58 (2008), no.~8, 931--941. 

\bibitem{RuzzaYang}
G.~Ruzza \& D.~Yang.
\newblock ``On the spectral problem of the quantum KdV hierarchy''.
\newblock {\em J. Phys. A} 54 (2021), no.~37, Paper No.~374001, 27~pp. 

\bibitem{Skoruppa}
N.-P.~Skoruppa.
\newblock ``A quick combinatorial proof of Eisenstein series identities''.
\newblock {\em  J. Number Theory} 43 (1993), no.~1, 68--73.

\bibitem{Teleman}
C.~Teleman.
\newblock ``The structure of 2D semi-simple field theories''.
\newblock {\em Invent. Math.} 188 (2012), no.~3, 525--588. 

\bibitem{Zagier123}
D.~Zagier.
\newblock ``Elliptic modular forms and their applications''. 
\newblock {\em The 1-2-3 of modular forms}, 1--103,
Universitext, {\em Springer, Berlin,} 2008. 

\bibitem{Zagier}
D.~Zagier.
\newblock ``Partitions, quasimodular forms, and the Bloch--Okounkov theorem''.
\newblock {\em Ramanujan J.} 41 (2016), no.~1-3, 345--368. 

\end{thebibliography}
\end{document}